\documentclass[12pt]{article}
\pdfoutput=1
\usepackage{amsmath,amssymb,amsfonts,amsthm,fullpage,setspace,xparse,xcolor,tikz,pgfplots,caption,subcaption,paralist,comment,nicefrac,ulem,titling,times}

\usepackage[top=1.25in,bottom=1.25in,left=1.25in,right=1.25in]{geometry}

\usepackage[longnamesfirst]{natbib} 

\newcommand\blfootnote[1]{%
  \begingroup
  \renewcommand\thefootnote{}\footnote{#1}%
  \addtocounter{footnote}{-1}%
  \endgroup
}

\newtheorem{theorem}{Theorem}
\newtheorem{proposition}{Proposition}
\newtheorem{lemma}{Lemma}
\newtheorem{definition}{Definition}
\newtheorem{assumption}{Assumption}
\newtheorem{example}{Example}

\theoremstyle{remark}
\newtheorem*{remark}{Remark}

\newcommand{\ind}[1]{\ensuremath{\mathbf{1}_{#1}}}
\newcommand{\RR}{\mathbb{R}}

\DeclareDocumentCommand\Pr{ m g g }{\ensuremath{
    { \IfNoValueTF {#3}
        {   \IfNoValueTF {#2}
          {\mathbb{P}\left[{#1}\right]}
          {\mathbb{P}\left[{#1}\middle\vert{#2}\right]}
        }
        {\mathbb{P}_{#1}\left[{#2}\middle\vert{#3}\right]}
    }
}}
\DeclareDocumentCommand\E{ m g g }{\ensuremath{
    { \IfNoValueTF {#3}
        { \IfNoValueTF {#2}
          {\mathbb{E}\left[{#1}\right]}
          {\mathbb{E}\left[{#1}\middle\vert{#2}\right]}%
        }
        {\mathbb{E}_{#1}\left[{#2}\middle\vert{#3}\right]}
    }
}}

\DeclareMathOperator*{\argmax}{argmax}

\DeclareMathOperator\cx{conv}
\makeatletter
\renewenvironment{proof}[1][\proofname] {\par\pushQED{\qed}\normalfont\topsep6\p@\@plus6\p@\relax\trivlist\item[\hskip\labelsep\bfseries#1\@addpunct{.}]\ignorespaces}{\popQED\endtrivlist\@endpefalse}
\makeatother

\title{Optimal Disclosure of Information to Privately Informed Agents}
\author{Ozan Candogan \and Philipp Strack}
\date{}

\begin{document}

\maketitle
\blfootnote{We want to thank Dan Adelman, Dirk Bergemann,
John Birge,
Alex Frankel,
Nicole Immorlica, 
Emir Kamenica,
Jacob Leshno,
Brendan Lucier,
Stephen Morris, 
Rad Niazadeh,
Lars Stole,
Alex Wolitzky, Kai Hao Yang,  and participants of the EC conference 
for helpful discussions and suggestions.
Ozan Candogan thanks the University of Chicago Booth School of Business for financial support. 
Philipp Strack gratefully acknowledges financial support by the Sloan Foundation.}

\begin{abstract}
We study information design when there are  multiple agents interacting in a game who are privately informed about their types.
Each agent's utility depends on \emph{all} agents' types and  actions, as well as  (linearly) on the state. 
The optimal mechanism asks agents to report their types and then sends a private action recommendation to each agent which depends on all reported types and the state.
We show that there always exists an optimal mechanism which is laminar partitional.
Such a mechanism partitions the state space for each type profile and 
recommends the same action profile for states that belong to the same partition element.
Furthermore, the convex hulls of any two partition elements are such that either one contains the other or they have an empty intersection.
In the single-agent case, each state is either perfectly revealed or lies in an interval in which 
the number of different signal realizations is at most the number of different types of the agent plus two.
A similar result is established for the multi-agent case.

We also highlight the value of screening: without screening the best achievable payoff could be as low as one over the number of types fraction of the optimal payoff. 
Along the way, we shed light on the solutions of optimization problems over distributions subject to a mean-preserving contraction constraint and additional side constraints, which might be of independent interest.
\end{abstract}

\thispagestyle{empty}

\newpage
\setcounter{page}{1}

\section{Introduction}

We study how a designer can use information about a real-valued state to influence the belief and actions of a group of agents who possess private information.
For example, the agents could be competing firms who each decide on the quantity they produce.
The production costs could be each firm's private information and the state could measure the total demand for the product. 

The designer can without loss restrict attention to direct recommendation mechanisms where each agent truthfully reports his type  and then privately observes an action recommendation.
We prove that there always exists an optimal such mechanism with a particularly simple structure:
For each type profile there is a partition of the state space such that
the mechanism  recommends the same action profile for states that belong to the same partition element.
Thus, there exists a (deterministic) function mapping the state and vector of types to 
action recommendations.
Furthermore, the partition is laminar. 
This  implies that the convex hulls of any two partition elements are either nested or they do not overlap.
As a result of the laminar structure, an optimal partition can be completely described by the collection of (smallest) intervals containing the states that induce each action profile recommendation.
This structure is valuable for tractability as it reduces the designer's optimization problem from an uncountably infinite one to an optimization problem over the end points of the aforementioned intervals. 

Finally, we provide a bound on the ``depth'' of  optimal laminar partitions.
In the single agent case the laminar partition structure has depth of at most $|\Theta|+2$,  where $\Theta$ is the set of types of the agent.
That is, the interval associated with an action recommendation overlaps with at most $|\Theta|+1$ other intervals (associated with different action recommendations).
This implies that either (i) a state is perfectly revealed, or (ii) it lies in an interval in which the distribution of the posterior means  admits at most a finite number mass points. 
In the multi-agent case a similar bound on the depth of the laminar partitions can be obtained if the number of possible actions is finite for each agent.
In contrast to the single-agent case, where the bound is independent of the number of actions, this bound depends quadratically on it.
This difference is driven by the fact that while in the single-agent case the action recommendation reveals the partition element in which the state lies;
this is not the case when there are multiple agents.

Given that the state space is a continuum, it is not a priori clear how to obtain the optimal mechanism in a tractable way.
To address this question, we focus on the finite action case.
We identify a transformation in the single agent case that leads to a finite-dimensional convex program (despite the states and the space of signals being uncountably infinite).
Similarly, in the multi-agent case we derive a finite-dimensional (though not necessarily convex) program.

Furthermore, we discuss some properties of the optimal mechanism:
Focusing on the single agent case, we prove that  restricting attention to mechanisms that do not screen the agent (and reveal the same information to all types) can be strictly suboptimal and in general achieve only a $1/|\Theta|$ share of the optimal value for the designer.
This is in contrast to \cite{kolotilin2017persuasion} and \cite{guo2019interval} who show that in the binary action single agent case there is no benefit to screening the agent.

Through an example we illustrate that unlike in classical mechanism design, ``non-local'' incentive compatibility constraints might bind in the optimal mechanism (even if the agent's utility is supermodular in his actions and type).
Finally, under the optimal mechanism the actions of different types need not be ordered for all states. For instance, there are states where the low and the high types take a higher action than the intermediate types.\footnote{This can be leveraged to show that ``nested'' information structures that are optimal in related information design settings with 2 actions are suboptimal \citep[see e.g.,][]{guo2019interval}.}

As a crucial step in obtaining our results, we study optimization problems over distributions, where the objective is linear in the chosen distribution, and a distribution is feasible if it satisfies (i) a majorization constraint as well as (ii) some linear side constraints.
We characterize properties of optimal solutions to such problems. In particular, we show that one can find optimal distributions that redistribute the mass in each interval where the majorization constraint does not bind, to at most $n+2$ mass points, where $n$ is the number of side constraints.
Moreover, there exists a laminar partition of the underlying state space such that the signal based on this laminar partition ``generates'' the optimal distribution.
Our main result is proven by decoupling the information design problem over type profiles into optimization problems under majorization and linear side constraints.  
 Given the generality of such optimization formulations, we suspect that our results may have applications beyond the information design problem studied in the paper. We discuss some immediate applications in Section~\ref{sec:discussion}.

%
%
\paragraph{Literature Review}
Following the seminal work by \cite{Kam11}, the literature on Bayesian persuasion studies how a designer can use information to influence the action taken by an agent.
This framework has proven useful to analyze a variety of economic applications, such as the design of grading systems\footnote{\cite{ostrovsky2010information,boleslavsky2015grading,ray2020}.}, medical testing\footnote{\cite{schweizer2018optimal}.}, stress tests and banking regulation\footnote{\cite{inostroza2018persuasion, goldstein2018stress, orlov2018design}.}, voter mobilization and gerrymandering\footnote{\cite{alonso2016persuading,kolotilin2020economics}.},  as well as various applications in social networks\footnote{\cite{candogan2017optimal,candogan2019persuasionEC}.}. %
 For an excellent survey of the literature see \cite{kamenica2019bayesian} and \cite{bergemann2017information}.

Initial papers focused on either the case of a single agent who possesses no private information or the case where the designer uses public signals \citep{Broc07,rayo2010optimal,Kam11,gentzkow2016rothschild}.
\cite{kolotilin2017persuasion} and \cite{guo2019interval} extend this baseline model by considering the single-agent case where the agent possesses private information about his preferences and chooses between \emph{two} actions.
Assuming that the agent's payoff is linear and additive in the state, \cite{kolotilin2017persuasion} show that it is without loss to restrict attention to ``public'' signals,
which do not screen the agent and induce the same signal realization regardless of the type of the agent.
\cite{guo2019interval} consider 
a general monotone utility of the designer and the agent, but maintain the assumption of binary actions.
They show that even though not every outcome that can be implemented with private signals can also be implemented with public signals, it is nevertheless true that the \emph{designer-optimal} outcome can always be implemented with public signals. 
We complement this line of the literature by studying the case where the agent can potentially choose among \emph{more than two actions} and find in contrast with the binary action case that public signals could yield a payoff that is as low as one over the number of types fraction of the optimal one.\footnote{\cite{kolotilin2017persuasion} also provide an example showing that with more than $2$ actions restricting attention to public signals may result in a payoff loss (see online Appendix A of their paper). We strengthen this insight and in Section~\ref{subse:privatePublic} 
we establish that the maximal payoff loss due to focusing on public signals is one over the number of types fraction of the optimal one. Moreover, we show that this bound is tight.}


\cite{bergemann2013robust,bergemann2016bayes} consider information revelation to multiple agents and introduce the notion of ``Bayes correlated equilibria". 
Bayes correlated equilibria characterize the set of all outcomes that can be induced in a given game by revealing a private signal to each agent.
Thus, Bayesian persuasion problems can be solved by maximizing over the set of Bayes correlated equilibria.
While the basic concept does not allow for private information one can extend to the case with screening and private information \citep[see Definition 2 in][]{bergemann2017information}.
In this case, the private information is about the state, and hence an agent's payoff depends on his private information only through the state. 
As the designer learns the state once it is realized, she will be better informed about the agents' utilities than themselves.
While the formulation is present in the literature, as far as we know the structural properties of the optimal mechanisms are not well understood in the multiple agent case with private information.
In the present paper, we contribute to this literature in two ways:
First, we allow the utility of an agent to directly depend on his private information,
thereby relaxing the assumption that the designer is better informed than the agents  --
which might be economically restrictive in some settings.
Second, we consider a continuum of states and focus on quasi-linear utilities which allows us to   describe optimal mechanisms more explicitly in terms of laminar partitional signals.\footnote{\label{fot:linearity}Quasi-linearity assumption is commonly made in the literature. See for instance \cite{ostrovsky2010information,ivanov2015optimal,gentzkow2016rothschild,kolotilin2017persuasion,kolotilin2018optimal}. For a more detailed discussion of this setting and its economic applications see Section 3.2 in \cite{kamenica2019bayesian}.}

Without private information, the approaches in \cite{bergemann2016bayes,kolotilin2018optimal,dworczak2019simple},
can be used to characterize the optimal information structure.
These approaches lead to infinite-dimensional optimization problems even if there is a single agent with finitely many actions.
When there is a single agent, an alternative approach due to \cite{gentzkow2016rothschild} is to associate a convex function with each information structure, and cast the information design problem as an optimization problem over
all convex functions that are sandwiched in between two convex functions (associated with the full disclosure and no-disclosure information structures).
This also yields an infinite-dimensional optimization problem.
In contrast,  we provide a finite-dimensional optimization formulation that is applicable with multiple privately-informed agents and finitely many actions.
This formulation is also convex when there is a single agent, thereby providing a tractable framework for obtaining optimal mechanisms.

The aforementioned ``sandwiching'' constraint is equivalent to a majorization constraint restricting the set of feasible posterior distributions.
\cite{arieli2019optimal} and \cite{kleiner2020extreme} characterize the extreme points of this set.
As also observed in \cite{candogan2019optimality,candogan2019persuasionEC}, this characterization implies that in the single agent case without private information one can restrict attention to signals where each state lies in an interval such that for all states in that interval at most $2$ messages are sent. 
There are two additional critical challenges in our setting.
First, unlike earlier work, one needs to deal with additional constraints that stem  from the screening problem.
Second, since there are multiple agents, the information revealed to one agent can influence the actions taken by others,
which intricately couples the information design problems for different agents.
These challenges require a novel approach and render the information structures identified in the earlier literature suboptimal.


\section{Model}
We consider an information design setting in which a designer (she) tries to influence the action taken by privately informed agents (he/they), indexed by $i \in \{1,\ldots,|N|\} = N$.

\paragraph{States and Types}
We call the information controlled by the designer the state $\omega \in \Omega$ and the private information of agent $i$ his type $\theta_i \in \Theta_i$.
The state $\omega$ lies in an interval $\Omega = [0,1]$ and is distributed according to the (cumulative) distribution $F: \Omega \rightarrow [0,1]$, with density $f \geq 0$.\footnote{The assumption that the state lies in $[0,1]$ is a normalization that is without loss of generality
for distributions with bounded support as we can rescale the state (without affecting the linearity of the utility function imposed subsequently). 
Furthermore, while it is important that $F$ has no mass-points, all our result go through for any continuous distribution (which might not admit a density).}
Each agent's type $\theta_i$ lies in a finite set $\Theta_i$ and we denote by $\phi(\theta) > 0$ the probability that the type vector equals $\theta = (\theta_1,\ldots,\theta_{|N|}) \in \Theta \subseteq \prod_{i \in N} \Theta_i$.
We assume that the state $\omega$ and the types $\theta$ are independently distributed, but allow for arbitrary correlation between the types of different agents.

\paragraph{Signals and Mechanisms}

A \emph{direct mechanism} $\mu: \Theta \times \Omega \to \Delta(S)$ maps a type profile $\theta$ and a state $\omega$ to a conditional distribution $\mu^\theta( \cdot | \omega)$ over the set of signal realizations $S$.\footnote{Restricting attention to direct mechanisms is without loss of generality by the revelation principle.}
We denote by $\mu^\theta$ the signal\footnote{We follow the convention of the Bayesian persuasion literature and call a Blackwell experiment a signal.} associated with the type vector $\theta$, i.e.,
\[
    \mu^\theta(\cdot | \omega) = \Pr{s \in \cdot}{\omega,\theta}\,.
\]
Each signal realization $s \in S = \prod_{i \in N} S_i$ is $|N|$-dimensional.
The $i$-th coordinate $s_i$ is privately observed by agent $i$, but we allow for the signals observed by different agents to be correlated.
We restrict attention to signals for which Bayes rule is well defined,\footnote{Formally, this requires that $\Pr{\mu}{\cdot}{s}$ is a regular conditional probability.} and denote by $\Pr{\mu}{\cdot}{s}\in \Delta(\Omega)$ the posterior distribution induced over states by observing the signal realization $s$ in the mechanism $\mu$, and by $\E{\mu}{\cdot}{s}$ the corresponding expectation.
When there are finitely many signal realizations
\begin{equation}\tag{Bayes Rule}
    \Pr{\mu}{   \omega \leq x }{s} =  \frac{\sum_\theta \phi(\theta) \int_0^x \mu^\theta(\{s\}|\omega) dF(\omega) }{ \sum_\theta  \phi(\theta) \int_0^1 \mu^\theta(\{s\}|\omega) dF(\omega)} \,.
\end{equation}


\paragraph{The Agents' Actions and Utilities}
After observing his type $\theta_i$, each agent $i\in N$ reports it to the mechanism.
Given the reported  type profile $\theta$ and the state realization, the mechanism draws a signal from the corresponding distribution, and agent $i$  observes the $i$-th coordinate
 $s_i$
of the signal realization.
Then each agent $i$ chooses an action $a_i$ in a compact set $A_i$ to maximize his expected utility
\[
    \max_{a_i \in A_i}\E{u_i(a_i,a_{-i},\omega,\theta)}{ s_i, \theta_i } \,.
\]
We note that the expectation in the above expression is over the state $\omega$, the action taken by other agents $a_{-i}$, and their types $\theta_{-i}$.
If we impose  additional assumptions on the set of action profiles $A=\times_{i\in N} A_i$ we will explicitly mention them, and otherwise we allow it to be finite or infinite.

\paragraph{Recommendation Mechanisms} 

A \emph{direct recommendation mechanism} is a direct mechanism where the signal realization for each agent is an action recommendation, i.e., $S_i = A_i$.
A direct recommendation mechanism is incentive compatible if it is optimal for each agent to report his true type $\theta_i$ and follow the action recommendation, instead of (mis)reporting his type as $\theta_i' \in \Theta_i$ and choosing an optimal action afterwards.
Throughout  without loss we focus on incentive compatible direct recommendation mechanisms.
Formally, denoting by  $\sigma_i:A_i\rightarrow A_i$ the  action policy that maps an action recommendation to an action taken by  agent $i$, the incentive compatibility requirement can be stated as follows:\footnote{When $\theta_i=\theta_i'$ this constraint reduces to the \emph{obedience} constraint, which ensures that it is optimal for agent $i$ to follow the action recommendation.}
\begin{equation}\label{eq:opt-signal}
\begin{aligned}
    \sum_{\theta_{-i}} \phi( \theta ) &\int_\Omega \int_{A_{-i}} u_i(a_i,a_{-i},\omega,\theta ) d \mu^\theta(a | \omega ) d F(\omega) \\
    &\geq \max_{\sigma_i} \sum_{\theta_{-i}} \phi( \theta ) \int_\Omega \int_{A_{-i}} u_i(\sigma_i(a_i),a_{-i},\omega,\theta ) d \mu^{(\theta_i',\theta_{-i})}(a | \omega ) d F(\omega)
\end{aligned}
\end{equation}
for all $i$, $\theta_i,\theta_i'\in \Theta_i $.
One challenge in this environment is that each agent can deviate by simultaneously misreporting his type and taking an action different from the one 
that is recommended by the mechanism.

\paragraph{The Designer's Utility} 

We denote by $v: A \times \Omega \times \Theta \to \RR$ the designer's utility. 
For a given direct recommendation mechanism the designer's expected utility equals
\begin{equation}\label{eq:sender-payoff}
     \sum_{\theta} \phi( \theta ) \int_\Omega \int_{A} v(a,\omega,\theta ) d \mu^\theta(a | \omega ) d F(\omega) \,.
\end{equation}
The designer's information design problem is to pick a direct recommendation mechanism that satisfies \eqref{eq:opt-signal} to maximize \eqref{eq:sender-payoff}.

To make this setting with infinitely many states tractable we further focus on  preferences which are quasi-linear in the state:
\begin{assumption}[Quasi-Linearity]\label{ass:mean-belief}
    The agents' utilities 
    $\{u_i\}$
    and the designer's utility $v$ are quasi-linear in the state, i.e., for $i\in N$ there exist functions $u_{i1},u_{i2},v_1,v_2: A \times \Theta \to \RR$ 
    continuous in $a\in A$  such that
    \begin{align*}
        u_i(a,\omega,\theta) &= u_{i1}(a,\theta) \omega + u_{i2}(a,\theta)\\
        v(a,\omega,\theta) &= v_1(a,\theta) \omega + v_2(a,\theta) \,.
    \end{align*}
\end{assumption}
Assumption~\ref{ass:mean-belief} is natural in many economic situations and is commonly made in the literature (c.f. Footnote~\ref{fot:linearity}).\footnote{We also note that the continuity of the payoffs in $a$, and the compactness of $A$ together with Assumption \ref{ass:mean-belief} ensure that the payoffs are bounded, i.e., $|u_i(a,\omega,\theta)|, | v(a,\omega,\theta)| \leq B $ for some $B<\infty$. In what follows, this mild technical condition is used  to change the order of integrals that appear in the designer's and agents' problems.}
For example \cite{kolotilin2017persuasion} assume that there is a single agent who has  two actions $\{0,1\}$,
and that the agent's utility for one action is zero, and for the other action it is the sum of the type and state, which implies that $u_i(a_i,\omega,\theta) = a_i \times ( \omega + \theta )$.

\begin{remark}
Our results generalize to the case where the preferences of all agents and the designer depend linearly\footnote{In the single agent case we could allow $u,v$ to depend non-linearly on the agent's posterior expectation of $h(\omega)$.} on some (potentially) non-linear transformation of the state $h(\omega)$ 
as long as the distribution of $h(\omega)$ admits a density.\footnote{To see this note that for every function $h:\Omega \to \RR$ we can redefine the state to be $\tilde{\omega} = h(\omega)$.} What is crucial for our results is that the agents' belief about the state influences the preference of the designer and the agent only through the same real valued statistic.
\end{remark}

\subsection{A Motivating Example} \label{se:motivatingExample}

We  next provide an economic example to illustrate the model.
Two firms $1,2$ producing a good choose production quantities in $A_1=A_2=\{0,1,2\}$. 
The price of the product depends on the total production $a_1+a_2$ by the firms, and is given by $d - (a_1+a_2)$, where $d = (4+8 \omega)$ is the demand for the good and $\omega\sim U([0,1])$ is the state.
The unit production cost of each firm is its private type and equals $4$ or $6$ with equal probability independently of each other and the state (i.e., $\Theta = \{(4,4),(4,6),(6,4),(6,6)\}, \phi \equiv \nicefrac{1}{4}$).
The consumer surplus ($CS$) and total firm profits ($FP$) are respectively given by $CS=(a_1+a_2)^2/2$, and \[FP=
\big((4+8 \omega) - (a_1+a_2)\big)(a_1+a_2) -  a_1\theta_1-a_2\theta_2.\]

We are interested in characterizing the combinations of consumer surplus and firm profits that can be induced by a mediator that facilitates information exchange between the firms. %
To do so,
we numerically derive the information structures maximizing different weighted combinations of $CS$ and $FP$.\footnote{See Appendix \ref{app:optFormMultiAgent} for details on the numerical computations for this example.}
The results are illustrated in Figure \ref{fig:pareto}.

\begin{figure}[h!]
\begin{subfigure}[t]{.7\textwidth}
  \centering
  \caption{}
  \includegraphics[width=.9\linewidth]{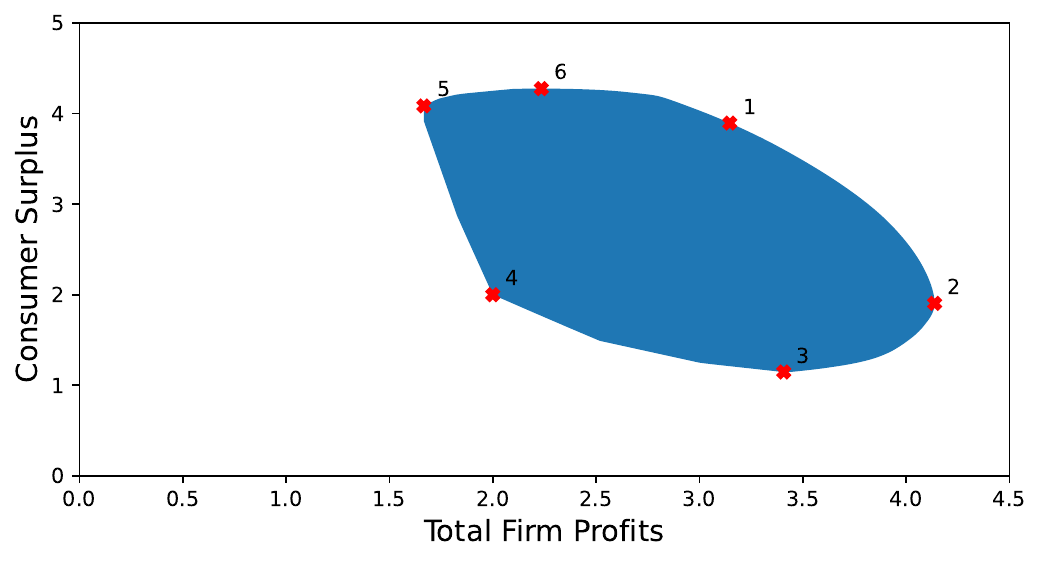}
  \label{fig:sub-first}
\end{subfigure}
\begin{subfigure}[t]{.29\textwidth}
  \centering
  \caption{}
 \raisebox{.8cm}{\includegraphics[width=\linewidth]{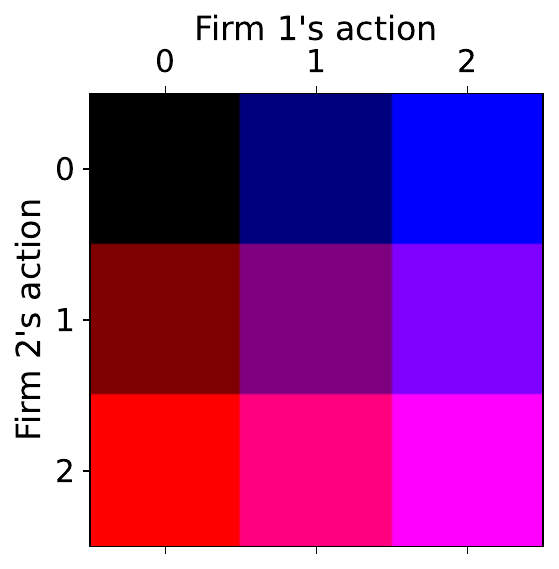}}
  
  \label{fig:sub-second}
\end{subfigure}
\newline
\vspace{.3cm}

\caption*{(c)}
\hspace{-.5cm}
  \begin{tabular}{c c}
        \includegraphics[width=.5\linewidth]{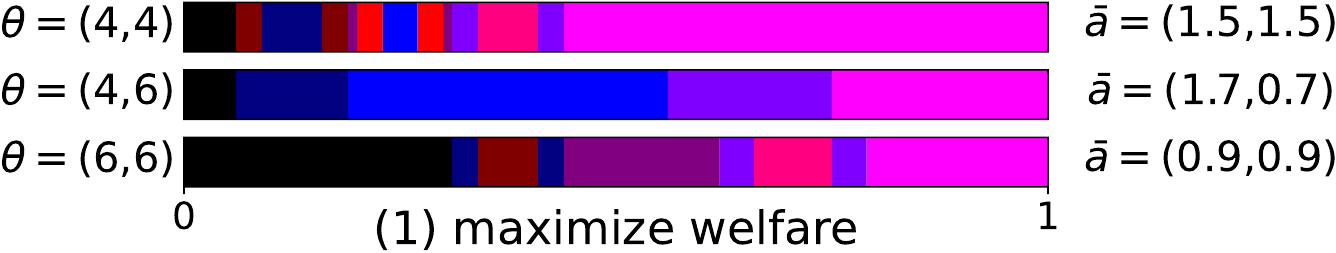}  & \includegraphics[width=.5\linewidth]{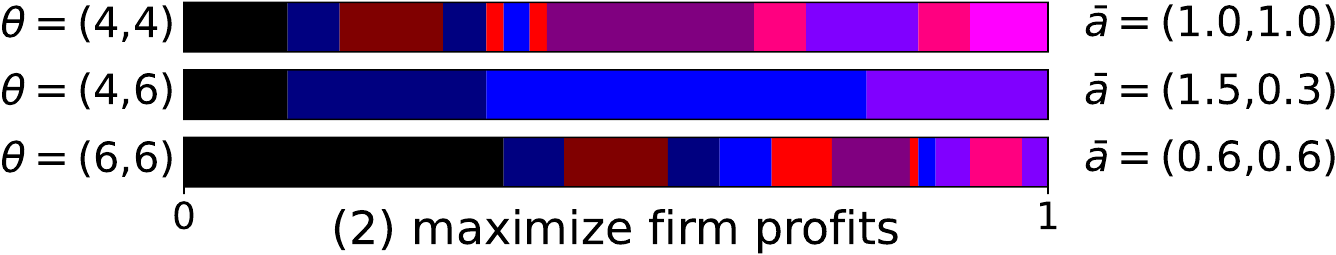}  \\
         \includegraphics[width=.5\linewidth]{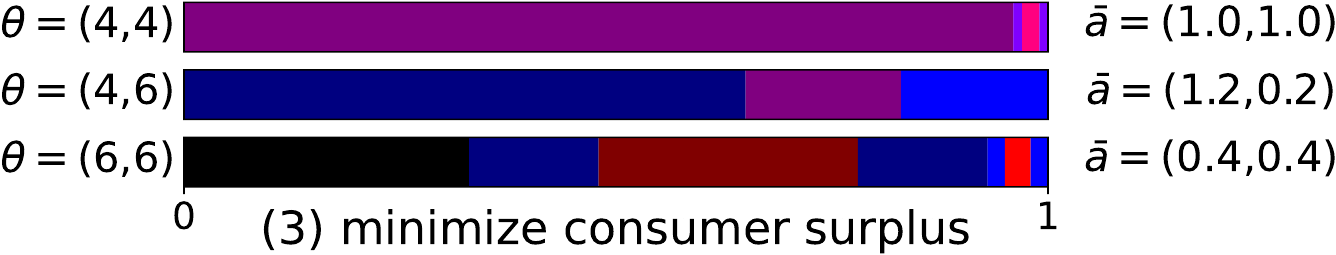}  & \includegraphics[width=.5\linewidth]{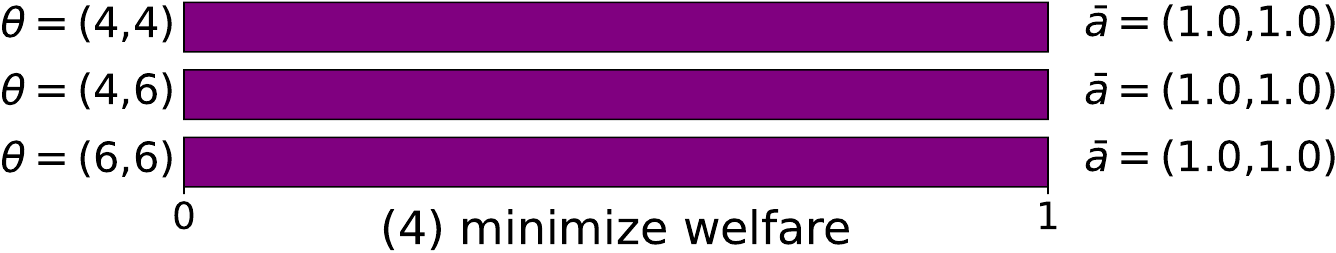}\\
         \includegraphics[width=.5\linewidth]{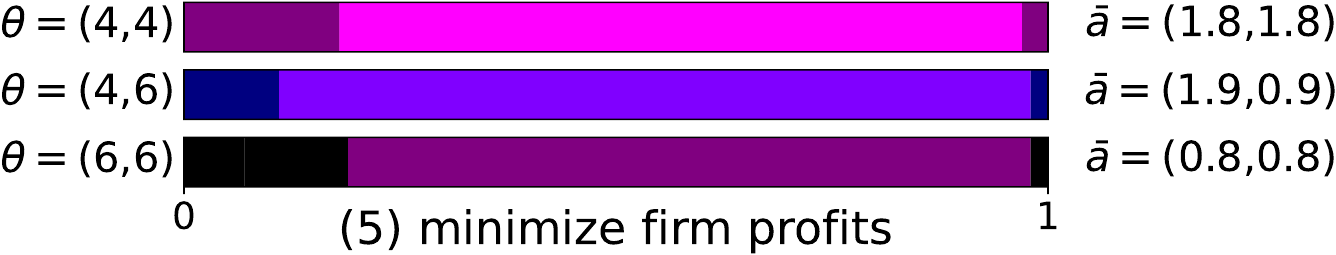}  & \includegraphics[width=.5\linewidth]{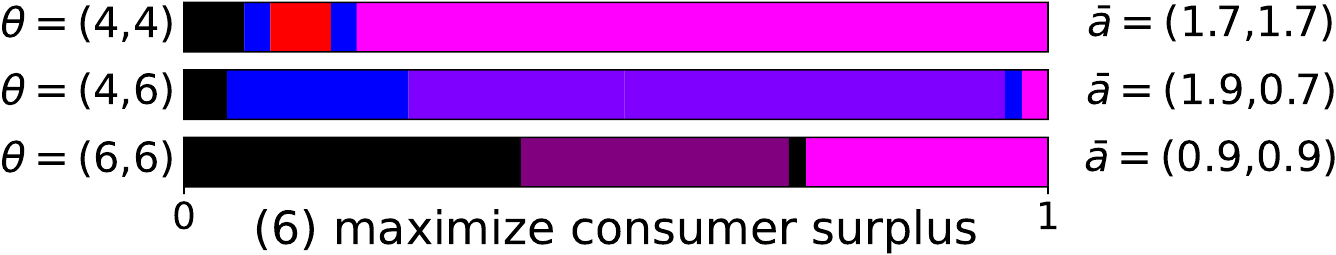}
    \end{tabular}  
    \label{fig:sub-third}
   \caption{(a) $CS$ and $FP$ achievable under different information structures. We highlight 6 points on this region that 
    achieve:
    (1) maximum welfare $CS+FP$,
    (2) maximum $FP$,
    (3) minimum $CS$,
    (4) minimum welfare $CS+FP$,
    (5) minimum $FP$,
    (6) maximum $CS$.
    (b) The colors assigned to different strategy profiles.
    (c) Optimal information structures. Here, 
    for any type profile $\theta$ we denote by
    $\bar{a}$ the vector of expected production quantities for both firms.
    }
    \label{fig:pareto}
    \end{figure}

We highlight $6$ points that are extremal in terms of achievable $CS$, $FP$, or welfare ($CS+FP$) and display the corresponding optimal information structures.
As our main result establishes we can restrict attention to simple (laminar) signals where for each type profile the state space is partitioned such that in each partition element the same actions are taken by the agents.
In this example this means that for each cost vector of the firms, the interval of possible demands is partitioned such that in each partition element the output vector of the firms is constant.
In Figure~\ref{fig:pareto} (b) and (c), we associate with each strategy profile a color, and use them to present the information structures that achieve these extremal points.
Given the symmetry between firms, the 
strategy profiles
associated with type profiles $(4,6)$ and $(6,4)$ are  same up to a permutation of the agents' identities. To avoid redundancy we only display one of them.

A few economic observations are worth highlighting:
First, when firms' costs are lower the expected production quantities are higher. 
While this monotonicity holds when taking the expectation over demand levels, production quantities are not monotone in the production cost for a fixed demand. For example when maximizing the $CS$ in (6) for some of the high demand levels the total production is higher for the production costs $(6,6)$ than it is for the production cost $(4,6)$.
Interestingly, the worst information structure in terms of welfare (4) is when no information about the cost of their competitors and the demand is revealed to the firms. 
Conversely, maximizing welfare, $FP$ or $CS$ leads to nontrivial laminar partitions (1,2,6).
To maximize firms' profits (2) the information structure induces more extreme asymmetric outcomes 
(e.g., $(0,2)$ or $(2,0)$ where one firm produces two units and the other produces zero)
relative to consumer surplus maximizing information structures (where balanced outcomes such as $(1,1)$   become more common). This leads to a larger number of distinct 
signal realizations in case of profit maximization.
Finally, based on the information structure, the consumer surplus and firms' profits vary significantly and there is more than a factor of two between the smallest and largest values of the aforementioned quantities.

\section{Analysis} \label{se:analysis}

Our analysis proceeds in several steps. 
First, we show that given a direct recommendation mechanism the designer can achieve the same payoff by using what we refer to as state garbling recommendation (SGR) mechanism.
This reduction is consequence of our restriction to quasi-linear utilities and an auxiliary step in proving our main result.
Second, we show that  optimal SGR mechanisms can be characterized through problems which are decoupled across type profiles (but not across agents\footnote{Note that there is no similar decoupling across agents, due to the strategic interactions among them.}).
Each of the decoupled problems involves  optimization over posterior mean distributions under linear side constraints.
Third, we establish that  solutions to such problems can always be induced by 
constructing a laminar partition and pooling states according to that partition. 
Finally, this implies our main result that there exists an optimal mechanism which 
for each type profile constructs a laminar partition of the state space and 
recommends the same action profile for states that belong to the same partition element.

\subsection{State Garbling Recommendation Mechanisms}\label{se:sgrm}

An SGR mechanism is an 
incentive compatible direct recommendation mechanism that
for each type profile $\theta$
has the following structure:
\begin{itemize}
    \item[(i)] The designer chooses an auxiliary signal $\nu^\theta$ whose realization $m \in [0,1]$ equals the induced posterior mean, i.e. $\E{\nu^\theta}{\omega}{m}=m$.    
    \item[(ii)] For each realized posterior mean  she chooses a distribution over recommended action profiles such that no action profile is recommended with positive probability at two different posterior means.
\end{itemize}

We next argue that due to our assumption of quasi-linear utilities the restriction to SGR mechanisms is without loss.\footnote{Note that without the  restriction in (ii) the set of mechanisms described above would equal the set of direct recommendation mechanisms as the designer could always chose a fully revealing signal in  (i). Due to restriction (ii), SGR mechanisms constitute a subset of the direct recommendation mechanisms.}
We start with an arbitrary direct recommendation mechanism $\mu$.
Let $m_{a,\theta} = \E{\mu^\theta}{\omega}{a}$ denote the mean of an outside observer's posterior belief about the state after observing the action profile $a \in A$ being recommended given the type profile $\theta \in \Theta$.
Note that this posterior belief never becomes known to the agents as they neither observe the complete type profile nor the recommended action profile.
Define $G^\theta:[0,1] \to [0,1]$ to be the cumulative distribution of posterior means given the type profile $\theta$
\[
    G^\theta ( x ) = \mathbb{P}_{\mu^\theta}\left[{m_{a,\theta} \leq x} \right]\,.
\]
Define $q^\theta \in \Delta(A)$ to be the distribution over action profiles conditional on   type profile $\theta$, i.e.,
\[
    q^\theta( B ) = \int_0^1 \mu^\theta(B|\omega) dF(\omega),
\]
 for  $B\subseteq A$.
Let $q^\theta(\cdot|x) \in \Delta(A)$   be the distribution over action profiles conditional on the posterior mean $m_{a,\theta}$ associated with the action profile being equal to $x \in [0,1]$
\[
    q^\theta( B | x ) = \frac{\int_B \ind{m_{a,\theta} = x} \, dq^\theta(a)}{ \int_A \ind{m_{a,\theta} = x} \, dq^\theta(a) } \,.
\]

Consider the  mechanism defined by the above tuple $(G,q)$ where $\nu^\theta([0,x])=G^\theta(x)$ and $q=(q^\theta)_\theta$ is the distribution over actions conditional on the posterior mean.
In this mechanism, given the type profile $\theta$, the designer first draws a signal realization $m$ according to $G^\theta$, and then recommends an action profile according to $q^\theta(\cdot|m)$. 
We claim that this is a valid SGR mechanism.
Note that it is possibly different from the direct recommendation mechanism we started with.

In this  mechanism -- assuming agents follow action recommendations --
the expected payoff of the designer given the type profile $\theta$ 
satisfies
\begin{equation} \label{eq:designerPayoff}
\begin{aligned}
    \int_\Omega & \int_A v(a,\omega,\theta) d\mu^\theta(a|\omega) d F(\omega) = \int_{A} \E{\mu^\theta}{ v(a,\omega,\theta) }{a} \, dq^\theta(a)\\
    &= \int_{A}  v(a,\E{\mu^\theta}{\omega}{a},\theta)  \, dq^\theta(a) 
    =  \int_{A}  v(a,m_{a,\theta},\theta)  \, dq^\theta(a) \\
    &= \int_\Omega \int_A v(a,m,\theta) \, dq^\theta(a|m) \, d G^\theta(m) \,.
\end{aligned}
\end{equation}
The first equality leverages the boundedness of the payoffs and changes the order of integration.
The second one follows from the quasi-linearity of $v$, the third from the definition of $m_{a,\theta}$, and the forth from the definition of $q^\theta$.

Using the same argument it can be readily seen that for any reported and true type profiles $\theta',\theta$ such that $\theta'_{-i}=\theta_{-i}$ when agents other than $i$ follow their action recommendations we have
\begin{equation} \label{eq:SGRMpayoffs}
    \int_\Omega  \int_{A_{-i}} u_i(a_i',a_{-i},\omega,\theta) d\mu^{\theta'}(a|\omega) d F(\omega) =   \int_\Omega \int_{A_{-i}} u_i(a_i',a_{-i},m,\theta) \, dq^{\theta'}(a|m) \, d G^{\theta'}(m) \,,
\end{equation}
where the left (right) hand side is the expected payoff of agent $i$ from observing action recommendation $a_i$ and taking action $a_i'$ in the initial (new) mechanism.
Since for any action recommendation the payoffs of agents coincide under the two mechanisms, it follows that the mechanism defined by the $(G,q)$ tuple satisfies incentive compatibility, and hence is a valid  SGR mechanism. Together with \eqref{eq:designerPayoff}, this observation implies that the two mechanisms
also yield the same payoff to the designer, and it is without loss to restrict attention to SGR mechanisms.


The expected payoff expression in the right hand side of \eqref{eq:SGRMpayoffs}
 can be used to obtain a characterization of incentive compatibility of SGR mechanisms. 
 Specifically,
 the SGR mechanism defined by $(G,q)$ is incentive compatible
if and only if for all $i\in N, \theta_i \in \Theta_i$
\begin{equation} \tag{IC}
\label{eq:ICNew}
    \begin{aligned}
     \sum_{\theta_{-i} \in \Theta_{-i}}& \phi(\theta)  \int_\Omega \int_A u_i(a,m,\theta) \, dq^\theta(a|m) \, d G^\theta(m) \\
    &   \geq \max_{\sigma_i,\theta_i'} \sum_{\theta_{-i} \in \Theta_{-i}} \phi(\theta)
    \int_\Omega
    \int_A u_i(\sigma_i(a_i),a_{-i},m,\theta) \, dq^{(\theta_i',\theta_{-i})}(a|m) \, d G^{(\theta_i',\theta_{-i})}(m) \,.
    \end{aligned}
\end{equation}

\paragraph{Feasible Posterior Mean Distributions}
Given that the designer's payoff and the incentive compatibility constraint can be expressed in terms of the distributions over posterior means $G=(G^\theta)$ and distributions over action profiles $q=(q^\theta)$ conditional on posterior means,
it may be possible to reformulate the designer's problem in terms of these quantities.
A natural question is thus which distributions over posterior means the designer can induce using a signal. 
An important notion to address this question is  \textsl{mean preserving contractions} (MPC).
A distribution over states $H:\Omega \to [0,1]$ is an MPC of a distribution $\tilde{H}: \Omega \to [0,1]$, expressed as $\tilde{H} \preceq H$, if and only if for all $\omega$ 
\begin{equation}\tag{MPC}\label{eq:MPC}
        \int_\omega^1 H(z) dz \geq \int_\omega^1 \tilde{H}(z) dz,
\end{equation}
and the inequality holds with equality for $\omega=0$. 

To see that $F \preceq G^\theta$ is necessary for $G^\theta$ to be the distribution of the posterior mean induced by some signal note that for every convex function $h:[0,1] \to \RR$ we have that 
\[
    \int_0^1 h(z) dF(z) = \E{ h( \omega)  } = \E{ \E{\mu}{ h(\omega) }{s} } \geq \E{ h(\E{\mu}{ \omega }{s}) } = \int_0^1 h(z) dG^\theta(z) \,.
\]
Here, the second equality is implied by the law of iterated expectations and the inequality follows from Jensen's inequality.
Taking $h(z) = \max\{0, z - \omega\}$   then yields that $F \preceq G^\theta$.
This condition is not only necessary, but also sufficient, see, e.g.,
 \cite{blackwell1950comparison,blackwell1954girshick,rothschild1970increasing} and \cite{gentzkow2016rothschild} for an application to persuasion problems. 

\begin{lemma}\label{lem:feasibility}
There exists a signal that induces the distribution $G^\theta$ over posterior means if and only if $F \preceq G^\theta$.
\end{lemma}
This result readily implies that a vector of type profile dependent posterior mean distributions $(G^\theta)_{\theta \in \Theta}$ is feasible if and only if $F \preceq G^\theta$ for all $\theta \in \Theta$.

\paragraph{Optimal SGR Mechanisms}

Combining the characterization of incentive compatibility from \eqref{eq:ICNew} and feasibility from Lemma~\ref{lem:feasibility} we next provide a characterization of optimal SGR mechanisms.
\begin{proposition}\label{prop:optimal-mechanism}
An SGR mechanism defined by $(G,q)$ is incentive compatible  and maximizes the designer's payoff if and only if $(G,q)$ solve
\begin{equation} 
\label{eq:optNewForm}
\tag{OPT}
    \begin{aligned}
         \max_{G,q}& \quad \sum_{\theta \in \Theta} \phi(\theta)\int_\Omega \int_A v(a,m,\theta) \, dq^\theta(a|m) \, d G^\theta(m) \\
         s.t. & \quad\eqref{eq:ICNew} \quad \& \quad F \preceq  G^\theta \quad \forall \theta.
    \end{aligned}
\end{equation}
\end{proposition}

One of the main challenges in this optimization problem is that even for a fixed $q$ the incentive compatibility constraint induces a strong interdependence among the components of $G$, which makes it impossible to optimize over them separately.
This interdependence is a natural economic feature of the multi-agent problem with private information as the designer cannot pick the action recommendation she provides to one agent and type without taking into account the fact that this might give other agents and types incentives to deviate.

\subsection{Decoupling the Problem Across Type Profiles} \label{se:decoupling}

Despite these challenges, we are able to characterize the structure of the optimal SGR mechanisms.
Our approach involves decoupling the designer's problem into $|\Theta|$ sub-problems (one for each type profile $\theta$) each involving optimization over only a single MPC constraint and linear side constraints.
As the argument for doing so and the precise decomposition differ significantly in the single- and multi-agent cases we explain them separately.

\subsubsection{The Single Agent Case} \label{subsubse:singleAgent}

In this section we consider the single agent case $|N|=1$ and thus drop the subindex indicating the agent's identity.
We define $\bar{u},\bar{v}:\Omega \times \Theta \to \RR$ to be the agent's and designer's \textsl{indirect utility functions}, i.e. their utility at a given mean belief $m$ if the agent takes an optimal action\footnote{We note that the indirect utility $\bar{u}$ is convex in $m$.}
\begin{align}
    \bar{u} ( m, \theta ) &= \max_{a \in A} u(a, m , \theta) \\
    \bar{v} (m , \theta) &= \max_{a \in A(m,\theta)} v(a,m,\theta)  \label{eq:barvDef}\,,
\end{align}
where $A(m,\theta) = \argmax_{b \in A} u(b,m,\theta)$.
Since in an SGR mechansim no action is recommended at two different posterior means the agent can infer the posterior mean from the action recommendation.
As any action recommendation policy $q$ that satisfies \eqref{eq:ICNew} must always recommend an action that is optimal for the agent at that posterior belief\footnote{Formally, this means that $a \notin A(m,\theta) \Rightarrow q^\theta(a|m) = 0$.}, we can rewrite \eqref{eq:ICNew} as
\begin{equation} \label{eq:IC-n1}
    \begin{aligned}
      \int_\Omega   \bar{u}(m,\theta) d G^\theta(m) \geq \max_{\theta'} 
    \int_\Omega
    \bar{u}(m,\theta) d G^{\theta'}(m) \,.
    \end{aligned}
\end{equation}

Let $(G^{\ast},q^{\ast})$ be an optimal solution to the problem given in Proposition~\ref{prop:optimal-mechanism}.
We define the value $e_\theta$   type $\theta$  could achieve when deviating optimally from reporting his type truthfully%
\begin{equation}\label{eq:ctheta}
    e_{\theta} =  \max_{\theta'\neq \theta} \int_\Omega \bar{u}( m, \theta) d G^{\ast,\theta'}(m)\,.
\end{equation}
We also define $d_\theta$ to be the value the agent gets when reporting his type truthfully
\begin{equation}\label{eq:dtheta}
    d_\theta = \int_\Omega \bar{u}( m, \theta) d G^{\ast,\theta}(m) \,.
\end{equation}
We note that $e_\theta,d_{-\theta}$ 
do not depend on
$G^{\ast,\theta}$.
We can thus characterize $G^{\ast,\theta}$ by optimizing over $G^\theta$ while taking $(G^{\ast,\theta'})_{\theta'\neq \theta}$ as given. This leads to our next lemma. 
\begin{lemma}\label{lem:decoupled-n1}
Consider the single agent case and let $e,d$ be the constants associated with an optimal SGR mechanism $(G^\ast,q^\ast)$.
 Then $(H^\theta,(G^{\ast,\theta'})_{\theta'\neq \theta},q^\ast)$ is an optimal SGR mechanism if and only if for any type $\theta \in \Theta$
the distribution $H^\theta$ solves
\begin{align}
    \max_{H^\theta \succeq F} \quad &   \int_\Omega \bar{v}(s, \theta) dH^\theta(s) \label{eq:prog-obj-n1}\\
    \text{s.t.}  \quad &\int_\Omega \bar{u}( s, \theta ) d H^\theta(s) \geq e_{\theta}  \label{eq:prog-IC-1-n1}\\
     &\int_\Omega \bar{u}( s, \eta ) d H^\theta(s) \leq d_{\eta} \qquad \forall \, \eta \neq \theta \,.\label{eq:prog-IC-2-n1}
\end{align}
\end{lemma}
In this formulation we maximize the payoff the designer receives from type $\theta$    under  constraint \eqref{eq:prog-IC-1-n1}. This constraint ensures that   type $\theta$ does not want to deviate  and report to be another type.\footnote{By considering the optimal deviation  we reduced the number of  incentive constraints in \eqref{eq:ctheta} from $(|\Theta|-1)$ to~$1$.}
Similarly,   constraint \eqref{eq:prog-IC-2-n1} ensures that no other type wants to report his  type as $\theta$.
We note that \eqref{eq:prog-IC-1-n1} and \eqref{eq:prog-IC-2-n1} encode the incentive constraints given in \eqref{eq:IC-n1} in which $G^\theta$ appears.

\subsubsection{The Multi-Agent Case}

We next turn to the multi-agent case.
Without loss we normalize here the probability of each type profile to $\phi(\theta) = \nicefrac{1}{|\Theta|}$ to make the equations easier to read.%
\footnote{This is without loss of generality as 
given a problem instance with arbitrary $\phi(\theta)$
one can define a new utility $|\Theta|\,\phi(\theta)\, u_i(a,\omega,\theta)$ for each agent $i$ and the designer $|\Theta| \,\phi(\theta)\, v(a,\omega,\theta)$ which entails exactly the same incentives in the original and the new problem instances. This amounts to a change of measure from $\phi(\cdot)$ to the uniform measure.}
The main challenge relative to the single agent case is that in an SGR mechanism the agents are in general unable to infer the posterior mean $m_{a,\theta}$ from their action recommendation.
As a consequence, the action recommendations $q$ are not determined by the \eqref{eq:ICNew} constraint and we cannot omit them from the problem.

Let  $(G^\ast,q^\ast)$ be a solution to \eqref{eq:optNewForm} and consider the corresponding SGR mechanism.
Define  $e_{i, \theta_i,\theta_i', \sigma_i}$ to be the  payoff\footnote{
More precisely these quantities correspond to payoffs multiplied with the probability $\sum_{\eta_{-i}} \phi(\theta_i,\eta_{-i})$ of the event that agent $i$'s private type equals $\theta_i$. For the subsequent discussion, this normalization does not play a role. Thus, with some abuse of terminology, we refer to such quantities as payoffs.}
agent $i$ of type $\theta_i$ gets   by reporting his type as $\theta_i' $
(where possibly $\theta_i' = \theta_i$)
and then deviating according to action policy $\sigma_i$:
\[
e_{i, \theta_i,\theta_i', \sigma_i} =
\sum_{\theta_{-i} \in \Theta_{-i} }
\int_\Omega \int_A u_i(\sigma_i(a_i),a_{-i},m,\theta_i,\theta_{-i}) \, dq^{\ast,(\theta_i',\theta_{-i})}(a|m) \, d G^{\ast,(\theta_i',\theta_{-i})}(m).
\]
Letting $I(\cdot)$ denote the identity, the payoff of agent $i$ from truthfully reporting his type and following the action recommendation equals $e_{i, \theta_i,\theta_i, I}$.
Similarly, the best payoff he can achieve 
 after misreporting his type
 and possibly taking an action different from the recommended one
 equals
 \begin{equation}\label{eq:bestDev}
  e_{i, \theta_i} = \max_{\sigma_i, \theta_i' \neq \ \theta_i} e_{i, \theta_i,\theta_i', \sigma_i} 
 .  
 \end{equation}
 It will be convenient to decompose the payoff
 $e_{i, \theta_i,\theta_i', \sigma_i}$ into payoff when the reported type profile equals some $ \theta' \in \Theta$ and the sum of payoffs $\gamma_{i, \theta_i,\theta_i', \sigma_i} (\theta')$ from other type profiles:\footnote{The quantity $\gamma_{i, \theta_i,\theta_i', \sigma_i} (\theta')$ is  equivalently given by the summation defining  $e_{i, \theta_i,\theta_i', \sigma_i}$ after excluding the summand with $\theta_{-i} = {\theta}_{-i}'$.}
 \[
    e_{i, \theta_i,\theta_i', \sigma_i}=    \gamma_{i, \theta_i,\theta_i', \sigma_i} (\theta')
    +\int_\Omega \int_A u_i(\sigma_i(a_i),a_{-i},m,\theta_i, \theta_{-i}') \, dq^{\ast,\theta'}(a|m) \, d G^{\ast,\theta'}(m)\,.
\]
By definition for $\theta_i' \neq  \theta_i$ the quantities $e_{i, \theta_i}$ and $\gamma_{i,\theta_i, \theta_i',\sigma_i}(\theta')$ do not depend on the posterior mean distribution $G^{\ast,\theta}$.
We next  show that  we can  characterize $G^{\ast,\theta}$ by optimizing 
the posterior mean distribution chosen for type profile $ \theta$ while taking those for other type profiles  $(G^{\ast,\eta})_{\eta\neq \theta}$ as given.
\begin{lemma}\label{lem:decoupled}
Let $e,\gamma$ be the constants associated with an optimal SGR mechanism
$(G^\ast,q^\ast)$. 
Then $(H^\theta,(G^{\ast,\theta'})_{\theta'\neq \theta},q^\ast)$ is an optimal SGR mechanism if and only if for any type $\theta \in \Theta$ the distribution $H^\theta$ solves
\begin{align*}
     & \quad \quad \max_{H^\theta \succeq F} ~  \int_\Omega \int_A v(a,m,\theta) \, dq^{\ast,\theta}(a|m) \, d H^\theta(m) \label{eq:prog-obj}\\
    &\text{such that} \\
    ~ & 
    \gamma_{i,\theta_i,\theta_i,I}( \theta) + \int_\Omega \int_A u_i(a,m,\theta) \, dq^{\ast,\theta}(a|m) \, 
        d H^\theta(m)   \\
        &\qquad\geq 
         \gamma_{i,\theta_i,\theta_i,\sigma_i} (\theta)
        + \int_\Omega \int_A u_i(\sigma_i(a_i),a_{-i},m,\theta
) \, dq^{\ast,\theta}(a|m) \, 
        d H^\theta(m) 
  \nonumber && \forall i, \sigma_i\neq I  
   \\
\quad & \gamma_{i,\theta_i,\theta_i,I} (\theta) + \int_\Omega \int_A u_i(a,m,\theta) \,dq^{\ast,\theta}(a|m) \, d H^\theta(m)
        \geq e_{i,\theta_i}
   &&   \forall i,
        \\ 
&\gamma_{i, \eta_i, \theta_{i},\sigma_i}(\theta)+\int_\Omega \int_A u_i(\sigma_i(a_i),a_{-i},m,\eta_i,\theta_{-i}) \, dq^{\ast,\theta}(a|m) \, d H^\theta(m) \leq e_{i,\eta_i,\eta_i,I} && \forall i, \sigma_i, \eta_i\neq \theta_i
\,. 
\end{align*}

\end{lemma}
This lemma can be best explained through a simple thought experiment. Suppose that we modify the initial SGR mechanism
$(G^\ast,q^\ast)$ by replacing distribution $G^{\ast,\theta}$ with $H^\theta$.
The left hand side of 
the first (and second) constraint 
is the payoff of agent $i$ in the new mechanism after reporting his type truthfully and following the recommendation. The right hand side of the first constraint is the payoff achieved via truthful type report followed by a deviation to action policy $\sigma_i$.
Similarly, the right hand side of the second constraint is the maximal payoff agent $i$ can achieve in the new mechanism by misreporting his type.
Note that when he misreports his type as $\theta'_i\neq  \theta_i$, the signal is drawn from a distribution other than $H^\theta$.
Thus, the right hand side  of
the second constraint   is a constant in this problem.
Finally, the left hand side of the third constraint is the payoff agent $i$ can guarantee by misreporting his type as $\theta_i$, when his type is actually $\eta_i$. The right hand side is the payoff from truthful reporting.
When $H^\theta$ satisfies these constraints, it follows that the resulting  mechanism $(H^\theta,(G^{\ast,\theta'})_{\theta'\neq \theta},q^\ast)$  still satisfies incentive compatibility, and is a valid SGR mechanism. 
As the objective in this optimization problem is the designer's payoff for the type profile $\theta$, this implies the lemma.

\subsection{Laminar Partitional Signals}

We next describe a small class of signals, \textsl{laminar partitional signals}.
We first define partitional signals:
\begin{definition}[Partitional Signal]
A signal 
$\mu$
is partitional if for each signal realization $s \in S$ there exists a set $P_s \subseteq \Omega$ such that $\mu(\{s\}|\omega)=\mathbf{1}_{\omega \in P_s}$.
\end{definition}
A partitional signal partitions the state space into sets $(P_s)_s$ and reveals to the agent the set  in which  the state $\omega$ lies. 
Partitional signals are thus noiseless in the sense that the mapping from the state to the signal is deterministic. {A simple example  of signals which are not partitional are normal signals where the signal equals the state $\omega$ plus normal noise and thus is random conditional on the state.}
Denote by $\cx(\cdot)$ the convex hull.
The next definition further restricts   the partition structure.

\begin{definition}[Laminar Partitional Signal] 
A partition $(P_s)_s$ is laminar if there is a partial order 
$\rhd$ on $S$ such that
$P_s= \cx P_s \setminus \cup_{s'|s \rhd s'} \cx P_{s'}$ for any $s$.
A partitional signal is laminar if its associated partition is laminar.
\end{definition}
This definition readily implies that $\cx P_s \cap \cx P_{s'} \in \{ \emptyset , \cx P_s, \cx P_{s'}\}$ for any $s, s'$.
The restrictions imposed by laminar partitional signals are illustrated in Figure~\ref{fig:laminar}.
Note that the elements of the laminar partition may belong to disjoint intervals
(see Figure \ref{fig:pareto} for various examples of laminar partitions).
We define the depth of a laminar partition as the smallest number $k$ such that the state space can be partitioned into intervals each of which contains at most $k$ elements of the laminar partition and every partition element is contained in one interval.
Intuitively, this captures how complicated the laminar partition is:
If convex hulls of partition elements are disjoint (nested) the depth is equal to $1$ (the number of signal realizations~$|S|$).

\begin{figure}[t]
\centering
\begin{tabular}{c  c}
  \begin{tikzpicture}
\begin{axis}[
    axis y line=none,
    y=0.5cm/3,
    restrict y to domain=0:4,
    axis lines=left,
    enlarge x limits=upper,
    scatter/classes={
        o={mark=*,fill=white}
    },
    scatter,
    scatter src=explicit symbolic,
    every axis plot post/.style={mark=none,line width=4pt},
    legend style={
        draw=none,
        at={(1,2)},
        anchor=south east
    },
    legend image post style={mark=none}
]
\addplot table [y expr=2,meta index=1, header=false] {
0.0 c
0.2 c

0.5 c
0.8 c
};\addlegendentry{$P_1$}
\addplot table [y expr=3,meta index=1, header=false] {
0.2 c
0.5 c

0.8 c
1.0 c
};\addlegendentry{$P_2$}
\end{axis}
\end{tikzpicture}   
\qquad &  \qquad
\begin{tikzpicture}
\begin{axis}[
    axis y line=none,
    y=0.5cm/3,
    restrict y to domain=0:4,
    axis lines=left,
    scatter/classes={
        o={mark=*,fill=white}
    },
    scatter,
    scatter src=explicit symbolic,
    every axis plot post/.style={mark=none,line width=4pt},
    legend style={
        draw=none,
        at={(1,2)},
        anchor=south east
    },
    legend image post style={mark=none}
]
\addplot table [y expr=2,meta index=1, header=false] {
0.0 c
0.1 c

0.4 c
0.6 c

0.9 c
1 c
};\addlegendentry{$P_1$}
\addplot table [y expr=3,meta index=1, header=false] {
0.1 c
0.4 c
};\addlegendentry{$P_2$}
\addplot table [y expr=3,meta index=1, header=false] {
0.6 c
0.9 c
};\addlegendentry{$P_3$}
\end{axis}
\end{tikzpicture}
\end{tabular}
\caption{\label{fig:laminar}The partition of the state space $\Omega=[0,1]$ on the left is not laminar while the partition on the right is laminar as the convex hull of all pairs of sets $P_1,P_2,P_3$ are either nested or have an empty intersection.}
\end{figure}
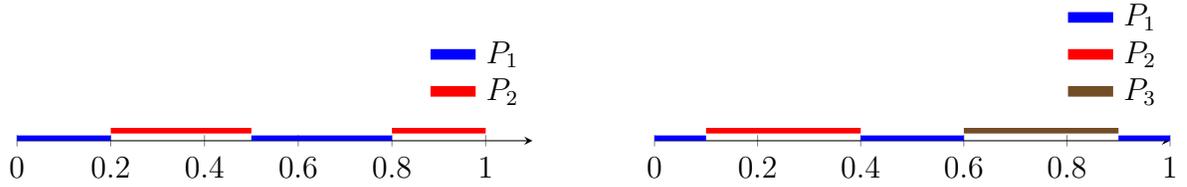

\begin{definition}[Laminar Partitional Mechanism]
A direct recommendation mechanism is laminar partitional if it consists of laminar partitional signals, i.e.,  for each type profile $\theta$ there exists a laminar partition $P^\theta$ of the state space $\Omega$ such that the same action profile is recommended in each partition element.
\end{definition}

We next establish that there always exists a laminar partitional mechanism that is optimal.
To simplify notation we denote by $P^\theta(\omega) = \{ P_s^\theta \colon \omega \in P_s^\theta \}$ the set of states where the same signal is realized   as in  state $\omega$, for a partitional signal with partition $P^\theta=(P^\theta_s)_s$.
\begin{theorem}\label{thm:optimal-signals}
Let $|A|$ be finite or $|N|=1$.
There exists an optimal laminar partitional mechanism.
Furthermore, given the partitions partition $P^\theta$ for each $\theta$ there exists  intervals $I^\theta_1,I^\theta_2,\ldots$ such that 
\begin{compactenum}[(i)]
    \item $\omega \notin \cup_k I^\theta_k$ implies $P^\theta(\omega) = \{\omega\}$;
    \item $\omega \in I^\theta_k$ implies   $P^\theta(\omega) \subseteq I^\theta_k$.
\end{compactenum}  
\end{theorem}

The proof is based on a result that characterizes the solutions to optimization problems over mean preserving contractions under linear side constraints, such as those in Section~\ref{se:decoupling}. As this result might be of independent  interest  we explain it in Section~\ref{sec:mpc-optimization}.

Theorem~\ref{thm:optimal-signals}  simplifies the search for optimal mechanisms. First, it implies that 
for each type profile $\theta$
the designer needs to consider only partitional signals, which (deterministically)
recommends the same action profile for all states in an element of the partition $P^\theta$.
Theorem~\ref{thm:optimal-signals} thus implies that the designer does not need to rely on random signals whose distribution conditional on the state could be arbitrarily complex. 

In general, the partitions that define a deterministic signal can be quite complex, and, for instance, each partition element can be a disjoint union of countably many subsets of states.
The fact that the partition can be chosen to be laminar is thus a further important simplification. 
To see why, consider the case with finitely many action profiles. 
Since the optimal signal is partitional the signal realizations correspond to at most $|A|$ subsets of the state space.
Due to the laminar structure each subset can be identified with its convex hull, which is an interval.
As each interval is completely described by its endpoints it follows that each laminar partitional signal can be identified with a point in $\RR^{2 |A|}$.
Thus, the problem of finding the optimal mechanism can be written as an optimization problem over $\RR^{2 |A| \times |\Theta|}$.
This contrasts with the space of mechanisms which can not be embedded in a finite dimensional Euclidean space even if one restricts to finitely many signal realizations.

The second part of the theorem implies that for each fixed type profile $\theta$ there are two types of partition elements: There are ``pooling intervals'' $(I^\theta_k)$ where multiple state realizations induce the same action profile recommendation.
In their complement, each state is mapped to a unique action profile recommendation.\footnote{Signals that induce such outcomes are  relevant  when there is a continuum of action profiles.}
Moreover, each  pooling interval 
can equivalently be expressed as  the union of partition elements it intersects, which also constitute a laminar partition of this pooling interval.
This implies that the task of constructing laminar partitions that induce  optimal posterior mean distributions also decouples over pooling intervals.
We next provide bounds on the depth of the laminar structure:

\begin{proposition} \label{prop:Thm2ndPart}
Consider the setting of Theorem \ref{thm:optimal-signals}.
\begin{compactenum}[(i)]
    \item If $|N|=1$, then 
    in each 
     $I^\theta_k$
     at most $|\Theta|+2$ action profiles are realized.
    \item If $|A|$ is finite, then
    $w\in \cup_k I^\theta_k$ almost surely, and
    in each 
     $I^\theta_k$ at most
     $\sum_{i \in N} |A_i|^2 |\Theta_i|+2$
     action profiles are realized.
\end{compactenum}
\end{proposition}
Suppose we restrict attention to
SGR mechanisms 
$(G,q)$, where $q^\theta(\cdot|m)$ 
is degenerate
and deterministically recommends an
  action profile for each posterior mean $m$ and type profile $\theta$
(which is without loss for the single agent case). Then Proposition~\ref{prop:Thm2ndPart} follows immediately by counting the number of side constraints in Lemma~\ref{lem:decoupled-n1} and \ref{lem:decoupled} which by Proposition~\ref{prop:general-result} given in the next section limit the depth of the laminar structure. To cover the case of
nondegenerate action profile distributions, an additional compactness argument is necessary. We provide this argument in the appendix.

Note that in Proposition \ref{prop:Thm2ndPart} part $(i)$ the number of action profiles that are realized in each interval is independent of the number of actions available to the agent, whereas this is not the case in part $(ii)$.
In fact, it is possible to construct numerical examples where the number of available actions impacts 
this quantity and the depth of the laminar partitional signals (see Appendix \ref{app:ZSexample}).

This dichotomy emerges since in the single-agent case, any action recommendation perfectly reveals to the agent the partition element containing the state and the corresponding posterior mean.
As explained in Section~\ref{subsubse:singleAgent} this implies that one can completely express the problem in terms of indirect utilities which makes this equivalent to a problem without action choices.
In contrast, in the multi-agent case due to the uncertainty about other agents' types and action recommendations the partition element does not become common knowledge among the agents. 
As a consequence, it is not possible to express the designer's problem in terms of indirect utilities dropping the actions.

\subsection{Maximizing under MPCs and Side Constraints}\label{sec:mpc-optimization}

The next section derives an abstract mathematical result about optimization under MPC constraints
and side constraints that implies Theorem~\ref{thm:optimal-signals} and Proposition \ref{prop:Thm2ndPart}.
We discuss this result separately as similar mathematical problems emerge in economic applications other than Bayesian persuasion.\footnote{For example \cite{kleiner2020extreme} discuss how optimization problems under mean preserving contraction constraints naturally arise in delegation problems. We leave the exploration of other applications of this mathematical result for future work to keep the exposition focused on the persuasion problem.}

Consider the problem of maximizing the expectation of an arbitrary upper semicontinuous function $v:[0,1] \to \RR$ over all distributions $G$ that are mean-preserving contractions of a given distribution $F:[0,1] \to [0,1]$ subject to $n \geq 0$ additional linear constraints
\begin{equation}\label{eq:linear-problem}
    \begin{aligned}
        \max_{G \succeq F}& \int_\Omega v(s) d G(s)  \\
        s.t. & \int_\Omega u_i(s) d G(s) \geq 0 \text{ for } i \in \{1,\ldots,n\}\,.
    \end{aligned}
\end{equation}
Throughout, we assume that the functions $u_i:[0,1] \to \RR$ are continuous. 
The next result establishes conditions that need to be satisfied by any solution of  problem \eqref{eq:linear-problem}.
Our results extend the insights of \cite{candogan2019persuasionEC,arieli2019optimal,kleiner2020extreme} who analyzed the problem of maximizing over mean preserving contractions \emph{without} side-constraints. We allow for side constraints as they naturally appear as incentive constraints and in settings with multiple agents. While without side-constraints each interval is optimally contracted into a distribution with just two points in its support, we find that in general the cardinality of the support equals the number of side constraints plus two.

%

\begin{proposition}\label{prop:general-result}
There exists a solution $G$ to problem \eqref{eq:linear-problem} and a countable collection of disjoint intervals $I_1,I_2,\ldots$ such that $G$ equals  distribution $F$ outside the intervals, i.e., 
\begin{equation}\label{eq:GFCond1}
    G(x) = F(x) \text{ for } x \notin \cup_j I_j 
\end{equation}
and each interval $I_j=(a_j,b_j)$ redistributes the mass of $F$ among at most $n+2$ mass points $m_{1,j}, m_{2,j},\ldots, m_{n+2,j} \in I_j$
\begin{equation}\label{eq:GFCond2}
    G(x) = G(a_j) + \sum_{r = 1}^{n+2} p_{r,j} \mathbf{1}_{m_{r,j} \leq x} \quad \text{ for } x \in I_j
\end{equation}
with $\sum_{r=1}^{n+2} p_{r,j} = F(b_j) - F(a_j)$ and the same expectation $\int_{I_j} x dG(x) = \int_{I_j} x dF(x)$.
\end{proposition}

The existence of an optimal solution follows from standard arguments exploiting the compactness of the feasible set of  \eqref{eq:linear-problem}.
To establish the remaining claims  of Proposition~\ref{prop:general-result}, we first fix an optimal solution, and  consider an interval where the 
 MPC
 constraint does not bind at this solution. As both the constraints as well as the objective function in \eqref{eq:linear-problem} are linear functionals in the CDF we can optimize over (any subinterval of) this interval fixing the solution on the complement of this interval, to obtain another optimal solution.
In this auxiliary optimization problem the 
MPC
constraint is relaxed by a constraint fixing the conditional mean of the distribution on this interval.
This problem is now a maximization problem over distributions subject to the $n$ original constraints and an additional   identical mean constraint. It was shown in \cite{winkler1988extreme} that each extreme point of the set of distributions, which are subject to a given number $k$ of linear constraints, is the sum of at most $k+1$ mass points. 
For our auxiliary optimization problem, this ensures the existence of an optimal solution with $n+2$ mass points.
A challenge is  to establish that the solution to the auxiliary problem is feasible and satisfies the 
MPC constraint.
The main idea behind this step is to show that if it is not feasible, then one can construct an optimal solution where the MPC constraint binds on a larger set.
However, this can never be the case if we start with an optimal solution where the set on which the MPC constraint binds is maximal (which exists by Zorn's lemma).
Combining such an initial optimal solution with the optimal solution 
for the auxiliary optimization problem, we obtain a new solution that satisfies the conditions of the proposition over this interval.
By repeating this argument for all intervals where the MPC constraint does not bind it follows that the claim holds for the entire support.

\paragraph{Laminar Structure} Let $\omega$ be a random variable distributed according to $F$.
Our next result shows that each interval $I_j$ in Proposition \ref{prop:general-result} admits a laminar partition such that
when the realization of $\omega$ belongs to some $I_j$, revealing the  partition element that contains it
and simply revealing $\omega$ when it does not belong to any $I_j$ induces a posterior mean distribution, given by $G$.
Proposition \ref{prop:general-result}
together with this result 
yields the optimality
of partitional signals as stated in Theorem \ref{thm:optimal-signals},
 as well as the depth of the corresponding laminar families presented in Proposition \ref{prop:Thm2ndPart}.

\begin{proposition}\label{prop:laminar}
Consider the setting of Proposition \ref{prop:general-result} and let $\omega$ be distributed according to $F$.
For each interval $I_j$ there exists a laminar partition $\Pi_j=(\Pi_{r,j})_r$ such that for all $r\in \{1, \ldots,n+2\}$
\begin{equation}\label{eq:laminar-consistency}
    \Pr{ \omega \in \Pi_{r,j} } = p_{r,j} \,\,\,\,\,\,\text{ and }\,\,\,\,\,\, \E{ \omega }{ \omega \in \Pi_{r,j} } =  m_{r,j} \,.
\end{equation}
\end{proposition}

The proof of this claim relies on a partition lemma (stated in the appendix), which strengthens this result by shedding light on how the  partition $\Pi_j$ can be constructed. The proof of the latter lemma is inductive over the number of mass points. 
When $G$ given in Proposition \ref{prop:general-result} has two mass points in $I_j$, the partition element that corresponds to one of these mass points is an interval and the other one is the complement of this interval relative to $I_j$. Moreover, it can be obtained by solving a system of equations, expressed in terms of the end points of this interval, that satisfy   condition \eqref{eq:laminar-consistency}.
As this partition is laminar this yields the result for the case where there are only 2 mass points in $I_j$.

When $G$ consists of $k >2$ mass points in $I_j$ one can find a  subinterval, such that:
(i) the expected value of $\omega\sim F$ conditional on $\omega$ being inside this subinterval equals the value of the largest mass point, and
(ii) the probability assigned to the interval equals the probability $G$ assigns to the largest mass point. 
Conditional on $\omega$ being outside this interval, the distribution thus only admits $k-1$ mass points and is a mean preserving contraction of the distribution $F$.
This allows us to invoke the induction hypothesis to generate a laminar partition such that revealing in which partition element $\omega$ lies generates the desired conditional distribution of the posterior mean.
Finally, as this laminar partition combined with the subinterval associated with the largest mass point of $G$ in $I_j$ is again a laminar partition, we obtain the result for distributions consisting of $k>2$ mass points.

The proof of  Proposition \ref{prop:laminar} (and Lemma \ref{lem:construct} of the Appendix) details these arguments, and also offers an algorithm for constructing a laminar partition satisfying \eqref{eq:laminar-consistency}.
While the result is stated by focusing on the setting of Proposition \ref{prop:general-result}, as can be seen from the proof, the optimality of $G$ does not play any role. Hence, the claim continues to hold for any distribution $G$ that satisfies only conditions \eqref{eq:GFCond1} and \eqref{eq:GFCond2}.

 \section{Single Agent Case: Screening vs. No-Screening}\label{se:simulations}
 
 In this section, we focus on the single agent case $|N|=1$. Throughout we also assume that the set of actions is finite $|A|=\{1,\dots,|A|\}$ and the designer's payoff $v(a,\theta)$ depends only on the action and the agent's type. 
 Our setting thus reduces to the problem of persuading a privately informed agent, which is of independent interest. 
 The  case of binary actions was analyzed in \cite{kolotilin2017persuasion} and \cite{guo2019interval} (who analyze this problem under slightly different assumptions).
 We first show that in the single agent setting described above -- without restricting attention to binary actions -- the optimal (SGR) mechanism can be obtained by solving a  \emph{finite-dimensional convex program}  (Section \ref{subse:finiteAction}). Then, we exemplify the optimal mechanism and contrast it with the optimal mechanisms derived in the literature by restricting attention further to binary action settings  (Sections \ref{se:buyerExample} and \ref{subse:privatePublic}).
 
 \subsection{A Convex Program for the Single-Agent Case}\label{subse:finiteAction}


As a consequence of Assumption~\ref{ass:mean-belief} there exist a partition of $\Omega$ into intervals $(B_{a,\theta})_{a \in A}$ such that action $a$ is optimal for the agent of type $\theta$ if and only if his mean belief is in the interval $B_{a,\theta}$.
By relabeling the actions for each type we can without loss assume that the intervals $B_{a,\theta} = [b_{a-1,\theta},b_{a,\theta}]$ are ordered with respect to the actions,\footnote{Formally, $0 = b_{0,\theta} \leq b_{1,\theta} \leq \ldots \leq b_{|A|,\theta} = 1$. If an action $a$ is never optimal   for a type $\theta$ set $b_{a-1,\theta}=b_{a,\theta}=b_{|A|,\theta} = 1$. This is without loss as no signal induces a posterior belief of 1 with strictly positive probability and the action thus plays no role in the resulting optimization problem.}
and hence for all $m\in B_{a,\theta}$:
\begin{align*}
    \bar{u}(m,\theta) &= {u}_1(a,\theta)m+{u}_2(a,\theta)   \,. 
\end{align*}
%
%
Consider an SGR mechanism with posterior mean distributions $(G^\theta)$.
Denote by $p_{a,\theta}$ the probability that action $a$ is recommended to type $\theta$ and by $m_{a,\theta}  \in B_{a,\theta}$ the posterior mean induced by this recommendation. 
The expected payoff of  type $\theta$ from reporting his type as $\theta'$ equals
\begin{equation}\label{eq:IC1}
\begin{aligned} 
     \sum_{a' \in A } p_{a',\theta'} \, \bar{u}(m_{a',\theta'},\theta).
\end{aligned}
\end{equation}
Defining
\[
    z_{a,\theta} = m_{a,\theta} p_{a,\theta}
\]
to be the product of the posterior mean $m_{a,\theta}$ induced by the action recommendation  $a$ and the probability $p_{a,\theta}$ of that recommendation, the  incentive compatibility constraint \eqref{eq:IC-n1}  for type $\theta$
can be   expressed as:
\begin{equation}
    \label{eq:IC}
    \sum_{a \in A} 
    {u}_1(a,\theta)z_{a,\theta}+{u}_2(a,\theta)p_{a,\theta} 
\geq 
    \sum_{a' \in A } 
   \left[ \max_{a \in A} 
    {u}_1(a,\theta)z_{a',\theta'}+{u}_2(a,\theta)p_{a',\theta'}  \right] \quad \forall \theta'.
\end{equation}
Here, the left hand side is the payoff of this type from reporting his type truthfully and  subsequently following the recommendation of the mechanism, whereas the right hand side is the payoff from reporting type as $\theta'$ and taking the best possible action (possibly different than the recommendation of the mechanism) given the signal realization.
Recall that the distribution $G^\theta$  is an MPC of  $F$ for all $\theta$ (Lemma \ref{lem:feasibility}).
Our next lemma establishes that the MPC constraints also admit an equivalent restatement in terms of $( p , z )$.\footnote{This reformulation was first used in 
\cite{candogan2019persuasionEC}  and for completeness we include a proof in the online appendix.}
\begin{lemma} \label{lem:majorization}
 $G^\theta \succeq F$ if and only if $\sum_{a\geq \ell} z_{a,\theta } \leq \int_{1-\sum_{a\geq \ell} p_{a,\theta}}^1 F^{-1}(x)dx $, where the inequality holds with equality for $\ell=1$.
\end{lemma}

Our observations so far establish that  the incentive compatibility and MPC constraints can both be expressed in terms of the $(p,z)$ tuple.
As a consequence of these observations we can reformulate the problem of obtaining optimal  SGR mechanisms, given in Proposition~\ref{prop:optimal-mechanism}, in terms of $(p,z)$ as follows:
\begin{equation} \label{eq:optPrivate} \tag{OPT2}
	\begin{aligned}
		\max_{\substack{p  \in (\Delta^{|A|})^\Theta\\z \in \RR^{ |A| \times |\Theta| }\\ y \in \RR^{  |A| \times |\Theta|^2 } }} & \quad  \sum_{\theta \in \Theta} \phi(\theta) \, \sum_{a \in A} p_{a,\theta} v(a,\theta) \\
		s.t. \quad & \sum_{a\geq \ell } z_{a,\theta} \leq 
		\int_{1-  \sum_{a\geq \ell}p_{a,\theta} }^{1} F^{-1}(x) dx
		&&  \forall\,\theta\in \Theta,\ell > 1 , \\
&\sum_{{a\in A} } z_{a,\theta} = \int_{0}^{1} F^{-1}(x) dx &&  \forall\,\theta\in \Theta, \\
&  {u}_1(a,\theta)z_{a',\theta'}+{u}_2(a,\theta)p_{a',\theta'} \leq 
		y_{a',\theta,\theta'} &&\forall\, \theta,\theta' \in \Theta, a,a' \in A,\\
		&\sum_{a'\in A } y_{a',\theta,\theta'}    \leq 
		\sum_{a\in A} \left(
		 {u}_1(a,\theta)z_{a,\theta}+{u}_2(a,\theta)p_{a,\theta} \right)
		&&\forall\, \theta,\theta' \in \Theta,\\
	 &p_{a,\theta}    b_{a-1,\theta}   \leq z_{a,\theta}  \leq p_{a,\theta}   b_{a,\theta}    &&\forall\, \theta\in \Theta, a \in A\,. 
	\end{aligned}
\end{equation}
In this formulation, the first two constraints are the restatement of the MPC constraints (see Lemma~\ref{lem:majorization}).
The value $y_{a',\theta,\theta'}$ corresponds to the utility the agent of type $\theta$ gets from observing the signal associated with type $\theta'$ and taking the optimal action  when the recommended action is $a'$. It can be easily checked that $y_{a',\theta,\theta'} = \max_{a\in A} 
 {u}_1(a,\theta)z_{a',\theta'}+{u}_2(a,\theta)p_{a',\theta'}$ at an optimal solution.\footnote{This is because when  $y_{a',\theta,\theta' } $ is strictly larger than the right hand side, it can be decreased to construct another feasible solution with the same objective.}
Thus, it follows that the third and fourth constraints restate the incentive compatibility constraint \eqref{eq:IC}, by using $y_{a',\theta,\theta'}$ to capture the summands in the right hand side of the aforementioned constraint. 
Finally, the last constraint  captures that the posterior mean $z_{a,\theta}/p_{a,\theta}$ must lie in $B_{a,\theta}$ for the action $a$ to be optimal.

It is worth pointing out that \eqref{eq:optPrivate} is a finite-dimensional \emph{convex} optimization problem.
This is unlike the infinite dimensional optimization formulation of Proposition \ref{prop:optimal-mechanism}.
 \eqref{eq:optPrivate}  restates the designer's problem in terms of the $(p, z)$ tuple.
Two points about this reformulation are important to highlight. 
First, 
an alternative approach would involve optimizing directly over distributions
$G^\theta$ that satisfy the IC constraints \eqref{eq:IC-n1} and
 have a single mass point $m_{a,\theta}\in B_{a,\theta}$ for each $a\in A$ with weight $p_{a,\theta}$. 
This could be formulated as a finite-dimensional problem as well (by searching over the location $m_{a,\theta}$ and weight $p_{a,\theta}$ of each mass point).
However, this approach does not yield a convex optimization formulation as the set of such $(p, m)$ tuples is \emph{not} convex.
The formulation in \eqref{eq:optPrivate} amounts to a change of variables that yields a convex program.

Second, given an optimal solution to \eqref{eq:optPrivate}, the  distributions $(G^\theta)_{\theta \in \Theta}$ 
of an optimal SGR mechanism
can be obtained straightforwardly by placing a mass point with weight $p_{a,\theta}$ at $z_{a,\theta}/p_{a,\theta}$ for each action $a$ with $p_{a,\theta}>0$.
Moreover, as discussed in Section \ref{sec:mpc-optimization}, an optimal mechanism that induces these distributions can be obtained by constructing a laminar partition of the state space (by following the approach in Proposition \ref{prop:laminar} and Lemma \ref{lem:construct} of the Appendix).
These observations imply our next proposition.
 \begin{proposition}
 For every optimal solution $(p,z,y)$ of \eqref{eq:optPrivate} the SGR mechanism  which recommends the action $a$ for type $\theta$ with probability $p_{a,\theta}$ and induces a posterior mean of $z_{a,\theta}/p_{a,\theta}$  (when $p_{a,\theta}>0$) is an optimal mechanism. Moreover, there exists a laminar partitional mechanism implementing these distributions.
 \end{proposition}

\begin{remark}
For the multi-agent case it is possible to obtain a similar finite-dimensional optimization problem.
However, in this case, there are two difficulties. 
First, while in the single-agent case the actions associated with different posterior mean levels are known, this is not the case for multiple agents. 
This issue can be circumvented by optimizing over the order of posterior mean levels associated with different action profiles.
Second, unlike the formulation of this section the resulting optimization problem is non-convex. In some instances, including the one in Section~\ref{se:motivatingExample}, one can get around these difficulties by leveraging further structure of the problem. 
More generally, numerical methods for non-convex optimization can be used.
See Appendix~\ref{app:optFormMultiAgent} for details.
\end{remark}

 \subsection{An Example}\label{se:buyerExample}
 
 Section \ref{se:motivatingExample} illustrates optimal laminar partitional mechanisms in a Cournot game.
We next illustrate our results through a simpler single agent example. This example
generalizes the buyer-seller setting from \cite{kolotilin2017persuasion}, who assume single unit demand, to the case where the buyer can demand more than one unit and has a decreasing marginal utility in the number of units.
As our example reduces to their setup for the case of a single unit, this example allows us to highlight the effects of the buyer having more than two actions.

In this example, the agent is a buyer who decides how many units of an indivisible good to purchase.
He is privately informed about his type which captures his taste for the good.
The designer is a seller who 
controls information  about the quality of the good, captured by the state.
We assume that prices are linear in consumption and set the price of one unit of the good to $\nicefrac{10}{3}$.
The utility the buyer derives from the $a$-th unit of the good is given by
\[
    (\theta + \omega) \max\{ 5 - a , 0 \} \,.
\]
His marginal utility of consumption decreases linearly in the number of goods, increases in the good's quality $\omega$, and in his taste parameter $\theta$.
The quality of the good is distributed uniformly in $[0,1]$ and the buyer's taste parameter either takes a low $\theta = 0.3$, intermediate $\theta=0.45$, or high value $\theta = 0.6$ with equal probability. 
The seller commits to a
laminar partitional mechanism
to maximize the (expected) number of units sold.
It is straightforward to see that in this problem the agent considers finitely many actions: purchasing  $0,1$ and $2$ units (see Appendix \ref{app:detailsForBuyerExample}). Hence the designer's problem can be 
 formulated and solved using the finite-dimensional convex program of Section~\ref{subse:finiteAction}.
 We solve this program, and construct the optimal laminar partitional mechanism (displayed in Figure~\ref{fig:newEx})%
 .\footnote{In the figure,  the cutoffs are reported after rounding, e.g., the cutoff for the high type is approximately at $0.06$. For sake of exposition, in our discussion  we stick to the rounded values.}

\begin{figure}[h]
        \centering
        \includegraphics[width=15cm]{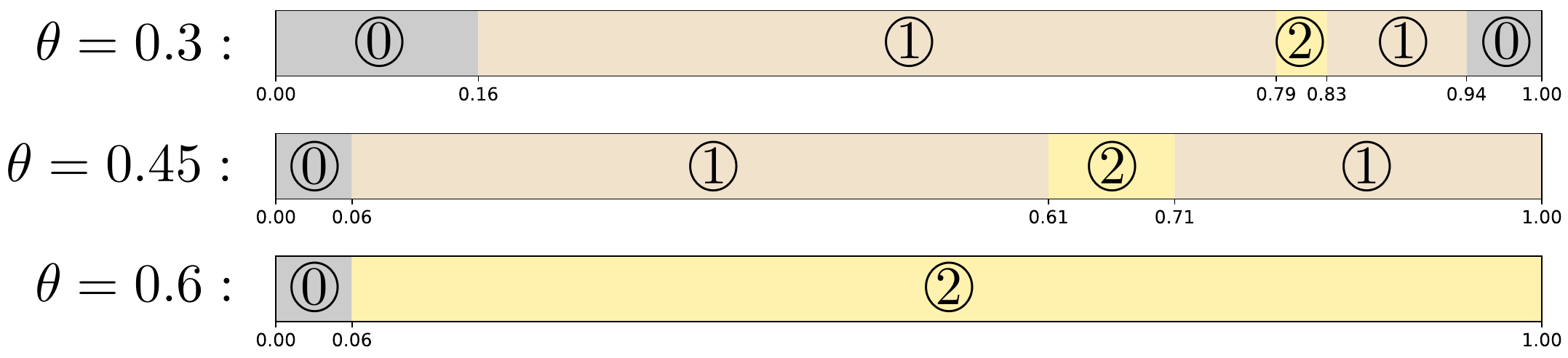}
        \caption{The optimal SGR mechanism.}
        \label{fig:newEx}
\end{figure}
In this figure,
each bar represents the state space and its differently colored regions the optimal partition (for the corresponding type).
For each type, the designer reveals whether the state belongs to the region(s) marked with $0,1,2$; and the buyer  finds it optimal to purchase the corresponding number of units.
Under the optimal mechanism, the \emph{expected} purchase quantity increases with the type.\footnote{This can be seen as the high type purchases two units in the states where the medium type purchases only one unit, which in turn leads to higher expected purchase. Similarly, when the low type purchases zero units, the medium type purchases zero or one units; and the size of the set of states where the medium type purchases two units is larger than that for the low type.}
While the expected quantities are ordered, the quantities purchased by different types for a given state  are \emph{not}. For instance, for   states between $0.79$ and $0.83$ the low and the high types purchase two units, and the medium type purchases one unit. 
Note that this implies that the purchase regions of buyers are not ``nested'' in the sense of \cite{guo2019interval}, who establish the optimality of such a nested structure for the case of two actions $|A| = 2$.
Moreover, low and medium types may end up purchasing lower quantities in some high states, than they do for lower states. In fact, under the optimal mechanism, for the best and the worst states, the low type purchases zero units.
Thus, in the optimal mechanism, the low and medium type of the buyer sometimes consume a \emph{smaller} quantity of the good if it is of higher quality.
This (maybe counterintuitive) feature of the optimal mechanism is a consequence of the incentive constraints:
By pooling some high states with low states, one makes it less appealing for the high type to deviate and observe the signal meant for a lower type.

\begin{remark}
In case of binary actions and under some assumptions on the payoff structure,%
\footnote{Both papers normalize the payoff of the action $0$ to zero. The assumption in \cite{kolotilin2017persuasion} is equivalent to the assumption that 
for $\theta' \leq \theta$
if $\E{u(1,\omega,\theta')} \geq 0$ then $\E{u(1,\omega,\theta)} \geq 0$ under any probability measure. \cite{guo2019interval} establish this result under assumptions that in our setting are equivalent to $\omega \mapsto \frac{u(1,\omega,\theta)}{v(1,\omega,\theta)}$  and $\omega \mapsto \frac{u(1,\omega,\theta)}{u(1,\omega,\theta')}$ are increasing for all $\theta' \leq \theta$.} 
\cite{kolotilin2017persuasion} and \cite{guo2019interval} establish that the optimal mechanism admits a ``public'' implementation. 
For each type the corresponding laminar partitional signal induces one action in a subinterval of the state space, and the other action in the complement of this interval. It can be shown that these intervals are nested which implies that the mechanism that reveals messages associated with different types to all agent types is still optimal. Thus, as opposed to first eliciting types and then sharing with each type the realization of the signal associated with this type, the designer can achieve the optimal outcome by sharing a  signal (which encodes the information of the signals of all types) publicly with all agent types. In other words, screening is not useful.
By contrast, it is straightforward to establish that the mechanism illustrated in Figure \ref{fig:newEx} does not admit a public implementation, and any public mechanism yields strictly lower payoffs  to the designer. See Appendix \ref{app:detailsForBuyerExample} for further details.
\end{remark}

\begin{remark}
Given the mechanism
of Figure~\ref{fig:newEx},
one can readily check  which incentive compatibility constraints are  binding.
It turns out that both the medium and the high types are indifferent among reporting their types as low, medium, or high.
Similarly, the low type is indifferent between reporting his type as low or medium, but achieves strictly lower payoff from reporting his type as high.
Interestingly, these observations imply that unlike in classical mechanism design settings ``non-local'' incentive constraints might bind in the optimal mechanism.\footnote{This in despite the fact that   the agent's utility is supermodular in his actions and type.}
\end{remark}

\subsection{The Value of Screening and Private Signals} \label{subse:privatePublic}
As discussed earlier,  the optimal laminar partitional mechanism reveals different signals to different types.
What if we restricted attention to public signals where all types observe the same signal? 
Suppose that the designer's payoff is non-negative.
For any  mechanism $(\mu^1, \ldots, \mu^{n})$ where different types observe different signals, the designer can always construct a public mechanism $(\mu^\theta, \ldots, \mu^\theta)$ where each type observes the signal $\mu^\theta$ associated with type $\theta$ in the original mechanism. 
Denoting by $G^\theta$ the posterior mean distribution under $\mu^\theta$, 
we conclude that
doing so and choosing $\theta$ optimally guarantees her at least a payoff of 
\[
    \max_{ \theta \in \Theta } \phi(\theta) \int_\Omega \bar{v} ( s , \theta  ) dG^\theta ( s ) \,.
\]
Since the designer's payoff is nonnegative, 
this is at least a $1/|\Theta|$ fraction of the payoff achieved by the original mechanism:
\[
    \sum_{\theta \in \Theta}\phi(\theta) \int_\Omega \bar{v} ( s , \theta  ) dG^\theta ( s ) \,.
\]
Thus, a public mechanism  guarantees a $1/|\Theta|$ fraction of the payoff achieved by the optimal 
 mechanism to the designer.
We next  establish that this bound is tight.
\begin{proposition} \label{prop:public-bound}Assume that the designer's utility $v$ is non-negative.
    \begin{compactenum}[(i)]
    \item In any problem there exists a public persuasion mechanism which achieves a $1/|\Theta|$ fraction of the optimal value achievable by an optimal mechanism. 
    
    \item In some problems no public persuasion mechanism yields more than a $1/|\Theta|$ fraction of the optimal value achievable by 
    an optimal mechanism.
    \end{compactenum}
\end{proposition}
We prove (the second part of) this proposition by explicitly constructing an example where the $1/|\Theta|$ ratio is achieved.
The idea behind the example is to give all types of the agent identical preferences and chose the payoff of the designer such that she wants different types of the agent to chose different actions.
In a public mechanism all agents have to choose the same action which leads to at most $1$ out of $|\Theta|$ types choosing the action preferred by the designer.
The example is constructed such that in a mechanism with private signals the designer can induce \emph{all types} to chose her most preferred action. If the payoff from inducing the correct action equals $1$ and the payoff from any other action to the designer equals $0$, this achieves the $1/|\Theta|$ bound. The main challenge in the construction, which is handled through a careful choice of payoffs, is to ensure that all types of the agent are indifferent between all signals to ensure that no type has incentives to misreport.

Two points are worthwhile highlighting about the example. 
First, it achieves the  worst case $1/|\Theta|$ bound  even when attention is restricted to a simple subclass of 
problem instances.
For instance,  the designer has a payoff of either $0$ or $1$ for different actions of the agent, and the agent has finitely many actions and type-independent utility functions. Second,  by relabeling the actions one can easily modify the example such that the designer's utility is independent of the agent's type and the agent's utility depends on his type. Proposition~\ref{prop:public-bound} thus holds unchanged even if one restricts attention to problems where the designer's utility depends only on the agent's action, but not on his type or belief.

\section{Discussion and Conclusion}\label{sec:discussion}

Our results can be extended in various dimensions. 
Persuasion problems where the designer's payoff depends on the induced posterior mean, but the admissible posterior mean distributions need to satisfy additional side-constraints are naturally subsumed.
Below we discuss some other economically-relevant extensions and  applications of our results.

\paragraph{Type-Dependent Participation Constraints}
In our analysis we can allow each type of an agent to face a participation constraint. That is, the mechanism must provide the relevant type  with at least some given expected utility. Our analysis and results carry over to this case unchanged as \eqref{eq:ICNew} already encodes such an endogenous constraint capturing the value of deviating by observing the signal meant for another type. To adjust the result for this case one just needs to in addition include the value of opting out of the mechanism in the incentive constraint.

\paragraph{Competition among Multiple Designers}

Another application of our approach is to competition among multiple designers. Suppose that each designer offers a mechanism and the agents can  choose to observe the signal of \emph{one} of them.\footnote{Another plausible model of competition is one where the agents can observe   the signals of   all designers. For an analysis of this situation see \cite{gentzkow2016competition}.} 
Each designer receives a higher payoff if an agent chooses 
her mechanism and might have different preferences over the agents' actions.
Again the designer has to ensure that the signal she provides each type of an agent with, yields a sufficiently high utility such that this type does not prefer to observe either another signal of the same designer or a signal provided by a different designer.
This situation corresponds to an endogenous type-dependent participation constraint which is determined in equilibrium.
As our analysis works for any participation constraint it also carries over to this case.

\paragraph{Beyond Persuasion Problems}

An immediate extension is to allow the designer to influence the agents' utilities by also designing transfers.
For instance, in the context of the example of Section \ref{se:buyerExample},  the seller might not only control the information she provides to the buyer, but also might charge different buyers different prices.
Such settings are  considered, e.g., in \cite{wei2020reverse,guo2018optimal,yang2020selling,yamashita2018optimal}.
As our results apply for any utility function, it is still without loss to restrict attention to laminar partitional signals.
Consider the case of a single agent, who (i) has finitely many actions, and (ii) his preferences are quasi-linear in the transfers.
The designer's optimal mechanism (which now determines the information structure as well as the transfers)
can be formulated following an approach similar to the one in Section~\ref{subse:finiteAction}.
Additional variables which capture transfers need to be added to the optimization formulation of that section.
Due to (i) these transfers can be represented by finite-dimensional vectors, and due to (ii) the resulting problem remains convex.
Thus, similar to  Section~\ref{subse:finiteAction} an optimal mechanism can be obtained tractably by solving a finite-dimensional convex program.

Finally, while this paper focused on persuasion problems, the mathematical result we obtain on maximization problems over mean preserving contractions under side-constraints can be applied in other economic settings which lead to similar mathematical formulations.
For example as first observed in \cite{kolotilin2019persuasion} the persuasion problem is closely related to delegation problems where the agent privately observes the state and the designer commits to an action as a function of a message sent by the agent. \cite{kleiner2020extreme} show that this problem can be reformulated as a maximization problem under majorization constraints which is a special case of the problem we discuss in Section~\ref{sec:mpc-optimization}. Our results thus allow one to analyze delegation problems where there is a constraint on the actions taken by the designer.\footnote{While mathematically closely related, the delegation problem is economically fundamentally different from the persuasion problem. For example the majorization constraint in the delegation problem corresponds to an incentive compatibility constraint while it corresponds to a feasibility constraint in the persuasion problem. The side constraints correspond to a feasibility constraint in the delegation problem while they correspond to an incentive compatibility constraint in the persuasion problem.} For example if the agent is the manager of a subdivision of a firm and the designer is the CEO who allocates money to that subdivision depending on the manager's report, our results allow one to analyze the case where the CEO faces a budget constraint and on average cannot allocate more than a given amount to that subdivision.

\appendix
\section*{Appendix}

\begin{proof}[Proof of Lemma~\ref{lem:majorization}]
The condition $G^\theta \succeq F$ can  equivalently be stated as:
\begin{equation} \label{eq:majorization1}
    \int_{0}^x (1- G^\theta(t)) dt \geq \int_{0}^x (1- F(t)) dt,
\end{equation}
for all $x$, where the inequality holds with equality for $x=1$. 
This inequality can be expressed in the quantile space as  
\begin{equation} \label{eq:majorization2}
    \int_{0}^x (G^\theta)^{-1}(t) dt \geq \int_{0}^x F^{-1}(t) dt,
\end{equation}
for all $x\in[0,1]$, with equality at $x=1$.
Note that since $G^\theta$ is a discrete distribution, this condition holds if and only if it holds for $x=\sum_{a\leq \ell} p_{a,\theta}$ and  $\ell \in A$.
For such $x$, we have
\begin{equation}
    \int_{0}^x (G^\theta)^{-1}(t)  = \sum_{a\leq \ell } p_{a,\theta} m_{a,\theta}= \sum_{a\leq \ell} z_{a,\theta},
\end{equation}
and \eqref{eq:majorization2} becomes
\begin{equation} \label{eq:majorization3}
\sum_{a\leq \ell} z_{a,\theta} \geq \int_{0}^{\sum_{a\leq \ell} p_{a,\theta}} F^{-1}(t) dt.
\end{equation}
Since $\int_{0}^1 F^{-1}(t) dt = \int_{0}^1 (G^\theta)^{-1}(t) dt = \sum_{a \in A} z_{a,\theta}$, the claim follows from \eqref{eq:majorization3} after rearranging terms.
\end{proof}

\begin{lemma}\label{lem:existence}
Let $|A|<\infty$ or $|N|=1$ an optimal mechanism exists. 
\end{lemma}
\begin{proof}
We first argue that an optimal mechanism exists in the case of finitely many actions $|A|<\infty$.
First, we note that set of feasible mechanisms is non-empty as the designer can always choose to reveal no information and induce a Bayes Nash equilibrium of the resulting game (which exists as there are finitely many types and actions).
The action recommendations of the associated direct mechanism simply recommend to each agent the action she would take knowing only her type in a  Bayes Nash equilibrium.
As we have argued in Section~\ref{se:sgrm} for every IC mechanism there exists an SGR mechanism which is IC and achieves the same payoff for the designer.
We can thus restrict attention to SGR mechanisms.
As discussed in Section~\ref{se:sgrm} these mechanisms are parametrized by $q^\theta \in \Delta(A)$ and $m_{a,\theta} \in [0,1]$.\footnote{Note here that $q^\theta(a)$ is the probability of the action profile $a$ given the type profile $\theta$ \emph{not} conditioning on the state.} Thus each SGR mechanism $(q,m)$ can be identified with a vector in $[0,1]^{2 |\Theta| |A| }$.
Furthermore, the expected utility of the designer and   an agent $i$ can be expressed respectively as
\begin{align*}
    &\sum_{\theta \in \Theta} \phi(\theta) \sum_{a \in A} q^\theta(a) u_i(a, m_{a,\theta},\theta), \\
    &\sum_{\theta_{-i}} \phi(\theta) \sum_{a_{-i}} q^\theta(\sigma_i(a_i), a_{-i}) u_i(\sigma_i(a_i), a_{-i} , m_{a,\theta},\theta) \,.
\end{align*}
Furthermore, the MPC constraint can be rewritten in the $(q,m)$ parametrization as
\[
    \sum_{a \in A} q^\theta(a) \max \{m_{a,\theta} - r, 0\} \geq \int_r F(x) dx \,
\]
for all $r \in [0,1]$ and with equality at $0$.
As both of the objective and the constraint are  continuous in $(q,m)$ it follows that the principal maximizes a continuous function over a compact subset of $[0,1]^{2 |\Theta| |A| }$ and hence a maximizer exists.


We next argue existence of a maximizer for the single agent case with an arbitrary action set.
As argued in Lemma~\ref{lem:compactness} the set of feasible distributions $G^\theta$ is sequentially compact.
As the product of finitely many sequentially compact spaces is also sequentially compact the set of vectors $(G^\theta)_{\theta \in \Theta}$ is also sequentially compact.
As $\bar{u}_i$ is continuous it follows that the IC constraint \eqref{eq:IC-n1} is continuous in $G$.
As $G \mapsto \sum_{\theta \in \Theta} \phi(\theta) \int_\Omega \bar{v}(s, \theta)  dG^\theta(s)$ is upper semicontinuous it follows that the designer maximizes an upper hemicontinuous linear function over a compact convex set.
By Bauer's maximum principle a maximizer exists. 

\end{proof}

\begin{lemma}\label{lem:compactness}
Suppose $u_i:[0,1]\rightarrow \mathbb{R}$ is a continuous function for $i\in  \{1,\ldots,n\}$.
The set of distributions $G:[0,1] \to [0,1]$ that satisfy $G \succeq F$ and 
\begin{align}\label{eq:compactnessConstraint}
        & \int_\Omega u_i(s) d G(s) \geq 0 \text{ for } i \in \{1,\ldots,n\}
\end{align}
is compact in the weak topology.
\end{lemma}

\begin{proof}[Proof]
First, note that as $u_i$ is continuous it is bounded on $[0,1]$.
Consider a sequence of distributions $G^k$, $k\in \{1,2,\ldots\}$ that satisfy the  constraints in \eqref{eq:compactnessConstraint}.
By Helly's selection theorem there exists a subsequence that converges pointwise. From now on assume that $(G^k)$ is such a subsequence and denote by $G^\infty$ the right-continuous representation of its point-wise limit.
Thus,   any sequence of random variables $m^k$ such that $m^k \sim G^k$ converges in distribution to a random variable distributed according to $G^\infty$.

As $(u_i)$ are continuous and bounded this implies that for all $i$ we have
 \[
     \lim_{k \to \infty} \int_\Omega u_i(s) d G^k(s) = \int_\Omega u_i(s) d G^\infty(s) \,.
 \]
Furthermore, for all $x \in [0,1]$
\[
    \lim_{k \to \infty} \int_x^1 G^k(s) ds = \int_x^1 G^\infty(s),
\]
and hence $G^\infty$ also satisfies $G^\infty \succeq F$.
We have hence established sequential compactness.
As the topology of weak convergence is metrizable by the Prokhorov metric and in a metrizable space, a subset is compact if and only if it is sequentially compact the set of distributions given in the statement of the lemma is  compact with respect to the weak topology.
\end{proof}
%
%
\begin{lemma}\label{lem:mass-preservation}
Let $F,G:[0,1] \to [0,1]$ be CDFs and let $F$ be continuous.
Suppose that $G$ is a mean-preserving contraction of $F$ and for some $x \in [0,1]$
\[
    \int_x^1 F(s) ds = \int_x^1 G(s) ds.
\]
Then $F(x) = G(x)$. Furthermore, $G$ is continuous at $x$.
\end{lemma}
\begin{proof}
Define the function $L:[0,1] \to \RR$ as $L(z) = \int_z^1 F(s) - G(s) ds \,.$
As $G$ is a mean-preserving contraction of $F$ we have that $L(z) \leq 0$ for all $z \in [0,1]$.
By the assumption of the lemma $L(x) = 0$.
By definition $L$ is absolutely continuous and 
has a weak derivative, which we denote
by $L'(z) = G(z)-F(z)$.
As $F$ is continuous $L'$ has only up-ward jumps and is right-continuous.
For $L$ to have a maximum at $x$ we need that $\lim_{z \nearrow x} L'(z) \geq 0$ and $\lim_{z \searrow x} L'(z) \leq 0$. This implies that 
\[
    \lim_{z \searrow x} G(z) - F(z) \leq 0 \leq \lim_{z \nearrow x} G(z) - F(z).
\]
In turn, this implies that $\lim_{z \searrow x} G(z) \leq \lim_{z \nearrow x} G(z)$. As  $G$ is a CDF it is non-decreasing and thus $G$ is continuous at $x$.
Consequently, $L$ is continuously differentiable at $x$ and as $L$ admits a maximum at $x$, we have that $0=L'(x)=G(x)-F(x)$.
\end{proof}

\begin{lemma}\label{lem:local-problem}
Fix an interval $[a,b] \subseteq [0,1]$, $c\in \mathbb{R}$, upper semicontinuous $v:[0,1] \to [0,1]$ and continuous $\tilde{u}_1,\ldots,\tilde{u}_n: [0,1] \to \RR$ and consider the problem 
\begin{align}
        \max_{\tilde{G}}& \int_\Omega v(s) d \tilde{G}(s) \label{eq:aux-objective} \\
        \text{ subject to }& \int_\Omega \tilde{u}_i(s) d \tilde{G}(s) \geq 0 \text{ for } i \in \{1,\ldots,n\}\label{eq:aux-side-constraint}\\
         & \int_{{a}}^{{b}} G(s) ds = c \label{eq:aux-mean-constraint}\\
         & \int_{[a,b]} d \tilde{G}(s) =1  \label{eq:aux-support-constraint}\,.
\end{align}
If the set of distributions that satisfy \eqref{eq:aux-side-constraint}-\eqref{eq:aux-support-constraint} is non-empty then there exists a solution to the above optimization problem that is supported on at most $n+2$ points.
\end{lemma}
\begin{proof}
Consider the set of distributions that assign probability $1$ to the set $[a,b]$.
The extreme points of this set are the Dirac measures in $[a,b]$.
Let $\mathcal{D}$ be the set of distributions which satisfy  \eqref{eq:aux-side-constraint}-\eqref{eq:aux-mean-constraint} and are supported on $[a,b]$.
By Theorem 2.1 in \cite{winkler1988extreme} each extreme points of the set $\mathcal{D}$ is the sum of at most $n+2$ mass points as \eqref{eq:aux-side-constraint} and \eqref{eq:aux-mean-constraint} specify $n+1$ constraints.
Note, that the set of the set of distributions satisfying \eqref{eq:aux-side-constraint}-\eqref{eq:aux-support-constraint} is compact.
As $v$ is upper semicontinuous the function $\tilde{G} \to \int_0^1 v(s) d \tilde{G}(s)$ is upper semi continuous and linear.
Thus, by Bauer's maximum principle (see for example Result~7.69 in \citealt{charalambos2013infinite}) there exist a maximizer at an extreme point of $\mathcal{D}$ which establishes the result.
%
\end{proof}

\begin{lemma}\label{lem:construct-feasible-sol}
Suppose that $H,G$ are distribution that assign probability 1 to $[a,b]$.
Let $M$ be an absolutely continuous function such that $\int_x^b G(s) ds > M(x)$ for all $x \in [a,b]$, and $\int_{\hat{x}}^b H(y) dy < M(\hat{x})$ for some $\hat{x} \in [a,b]$.
Then, there exists $\lambda \in (0,1)$ such that for all $x \in [a,b]$ 
\[
    \int_x^b (1-\lambda) G(s) + \lambda H(s) ds \geq M(x)
\]
with equality for some $x \in [a,b]$.
\end{lemma}

\begin{proof}
Define 
\[
    L_\lambda(x) = \int_x^b (1-\lambda) G(y) + \lambda H(y) dy - M(x)\,,
\]
and $\phi(\lambda) = \min_{z \in [a,b]} L_\lambda(z)$.
As $M$ is continuous, by the assumptions of the lemma we have that 
\[
    \phi(0) = \min_{x \in [a,b]} L_0(x) =  \min_{x \in [a,b]}  \left[ \int_x^b  G(s) ds - M(x)  \right] > 0 
\]
and
\[
    \phi(1) =  \min_{x \in [a,b]} L_1(x) =\min_{x \in [a,b]} \left[  \int_x^b  H(s) ds - M(x)\right] \leq \int_{\hat{x}}^b  H(s) ds - M(\hat{x}) < 0 \,.
\]
Furthermore, 
\[
    \left| \frac{\partial L_\lambda(z)}{\partial \lambda} \right| = \left| \int_x^b H(s) - G(s) ds \right| \leq b-a \,.
\]
Hence, $\lambda \mapsto L_\lambda(z)$ is uniformly Lipschitz continuous and the envelope theorem thus implies that
$\phi$ is Lipschitz continuous. As $\phi(0) > 0$, and $\phi(1) < 0$ there exist some $\lambda^* \in (0,1)$ such that $\phi(\lambda^*)=0$.
This implies that for all $x \in [a,b]$
\[
    \int_x^b (1-\lambda^*) G(s) + \lambda^* H(s) ds \geq M(x)
\]
with equality for some $x \in [a,b]$. This completes the proof. 
\end{proof}

{
\begin{lemma}\label{lem:Zorn}
For a solution $G$
to problem \eqref{eq:linear-problem}
denote the set of points where the mean preserving contraction constraint is binding by
\begin{equation}\label{eq:mpc-proof-main-result}
    B_G = \left\{ z \in [0,1] \colon \int_z^1 F(s) ds = \int_z^1 G(s) ds \right\}.
\end{equation}
There exists a solution to \eqref{eq:linear-problem} where the set $B_G$ is maximal in set inclusion sense.
\end{lemma}
\begin{proof}
As the set of feasible distributions is compact with respect to the weak topology by Lemma~\ref{lem:compactness} and the function $G \mapsto \int_0^1 v(s) dG(s)$ is upper semicontinuous  in the weak topology the optimization problem \eqref{eq:linear-problem} admits a solution. Observe that for any solution $G$, \eqref{eq:mpc-proof-main-result} implies that $B_G$ is a closed set.

Let ${\cal G}'$ denote the set of all solutions to \eqref{eq:linear-problem}. 
Denote by ${\cal G}$ the subset of ${\cal G}'$ such that (i)
$\{B_G | G\in {\cal G}\} = \{B_G | G\in {\cal G}'\}$ and
$B_{G_1} \neq B_{G_2}
$ for any $G_1,G_2\in {\cal G}'$ such that $G_1\neq G_2$ (the existence of such ${\cal G}$ follows from the axiom of choice).
Define a partial order $\succeq$ on $\cal G$:
$G_1 \succeq G_2$ if $B_{G_1} \supseteq B_{G_2}$.
Consider a totally ordered subset ${\cal G}_c$ of ${\cal G}$. 

Let $B = \cup_{G\in {\cal G}_c} B_G$ and denote by $\bar{B}$ the closure of this set. 
 Since $\bar{B}$ is closed and bounded, it is compact.
 Similarly $B_G$ is compact for each $G\in {\cal G}$.
 For $r\in \mathbb{N}_{+}$, consider a solution
 $G_r \in {\cal G}_c$ such that
 \[\max_{y\in \bar B} \min_{x\in B_{G_r}} |x-y|<1/r,
 \]
 where the optima are achieved due to compactness, the continuity of the argument being optimized (and the theorem of maximum).
 Existence of such a solution follows from the definition of $B$ and the fact that ${\cal G}_c$ is  totally ordered.
  By Helly's selection 
theorem, the sequence 
$\{G_r\}$ has a convergent subsequence. Let $G_\infty$ denote its limit. Since $G \mapsto \int_0^1 v(s) dG(s)$ 
is
 upper semicontinuous  in the weak topology,
 it follows that $G_\infty$  also solves \eqref{eq:linear-problem}.
 Furthermore, by construction, 
 $B_{G_\infty}$ is dense in $\bar B$.
 Since $B_{G_\infty}$ is also closed, it follows that
 it follows that $B_{G_\infty}=\bar{B}$. This implies that ${G_\infty} \succeq G$ for every $G\in {\cal G}_c$. 
 Zorn's lemma
 (see for example Section 1.12 in \citealt{charalambos2013infinite})
   implies that  ${\cal G}$ has a maximal element, $G^\star$. By the definition of our partial order this implies that  $B_{G^\star}\supseteq B_G$  for every $G\in {\cal G}$ and the claim follows.
\end{proof}
}

\begin{proof}[Proof of Proposition~\ref{prop:general-result}]
The first part of the claim follows from Lemma \ref{lem:Zorn}. The lemma also implies that there exists a solution $G$ for which the set $B_G$ defined in \eqref{eq:mpc-proof-main-result} is maximal in set inclusion sense. Consider such a solution.



Fix a point $x \notin B_G$.
We define $(a,b)$ to be the largest interval such that the mean-preserving contraction constraint does not bind on that interval for the solution $G$, i.e.
\begin{align*}
    a = \max \Big\{ z \leq x \colon z \in B_G \Big\} \qquad\qquad b = \min \Big\{ z \geq x \colon z \in B_G \Big\}.
\end{align*}
If $G$ assigns probability zero to the interval $[a,b]$ there are $0$ mass-points in the interval and we have thus established that there are less than $n+2$ mass-points in that interval.
Thus, assume for the rest of the proof that $G$ assigns strictly positive mass to $[a,b]$.
By Lemma~\ref{lem:mass-preservation} $G$ assigns no mass to $a$ or $b$ and hence $G$ also assigns strictly positive mass to the interior of $[a,b]$.
Consider now an interval $[\hat{a},\hat{b}] \subset (a,b)$ such that $G$ assigns strictly positive mass to $[\hat{a},\hat{b}]$.
We define $G_{[\hat{a},\hat{b}]}:[0,1] \to [0,1]$ to be the CDF of a random variable that is distributed according to $G$ conditional on the realization being in the interval $[\hat{a},\hat{b}]$
\[
    G_{[\hat{a},\hat{b}]}(z) = \frac{G(z)-G(\hat{a}_-)}{G(\hat{b})-G(\hat{a}_-)} \,,
\]
where $G(\hat{a}_-) = \lim_{s \nearrow \hat{a}} G(s)$. 
We note that $G_{[\hat{a},\hat{b}]}$ is non-decreasing, right-continuous, and satisfies $G_{[\hat{a},\hat{b}]}(\hat{b})=1$. Thus, it is a well defined CDF supported on $[\hat{a},\hat{b}]$.
As $G$ is feasible we get that
\begin{equation}\label{eq:local-constraint}
    \int_{\hat{a}}^{\hat{b}} u_k(s) d G_{[\hat{a},\hat{b}]}(s)  + \frac{1}{G(\hat{b})-G(\hat{a}_-)} \int_{\Omega\setminus [\hat{a},\hat{b}]} u_k(s) d G(s) \geq 0 \qquad \text{ for } k \in \{ 1,\ldots,n \} \,.
\end{equation}
To simplify notation we define the functions $\tilde{u}_1,\ldots,\tilde{u}_n$, where for all $k$ 
\begin{equation} \label{eq:uTilde}
    \tilde{u}_k(z) = u_k(z) + \frac{1}{G(\hat{b})-G(\hat{a}_-)} \int_{\Omega\setminus [\hat{a},\hat{b}]} u_k(y) d G(y) \,.
\end{equation}
Note that using this notation  \eqref{eq:local-constraint} can be restated as:
\begin{equation} \label{eq:uTilde2}
 \int_\Omega \tilde{u}_k(s) d G_{[\hat{a},\hat{b}]}(s) \geq 0 \qquad \text{ for } k \in \{1,\ldots,n\}.
\end{equation}
As $G$ satisfies the mean-preserving contraction constraint relative to $F$, using the fact that $a<\hat a$ and $\hat b<b$,
for $z\in [\hat a, \hat b]$
we obtain:
\begin{equation} \label{eq:GabStrict}
    \int_z^{\hat{b}} G_{[\hat{a},\hat{b}]}(s)  ds > \frac{1}{G(\hat{b}) - G(\hat{a}_-)} \left[\int_z^1 F(s) ds - \int_{\hat{b}}^1 G(s) ds - (\hat{b}-z) G(\hat{a}_-) \right] = M(z) \,.
\end{equation}
Consider now the maximization problem over distributions supported on $[\hat{a},\hat{b}]$
that satisfy  the  constraints  derived above (after replacing the strict inequality in \eqref{eq:GabStrict} with a weak inequality)
and maximize the original objective:
\begin{equation}
\begin{aligned}\label{eq:linear-problem-no-feasibility}
        \max_{H}& \int_\Omega v(s) d H(s)  \\
        \text{ subject to }& \int_\Omega \tilde{u}_i(s) d H(s) \geq 0 &&\text{ for } i \in \{1,\ldots,n\}\\
         & \int_z^{\hat{b}}  H(s) ds \geq M(z) &&\text{ for } z \in [\hat{a},\hat{b}]\\
         & \int_{[\hat{a},\hat{b}]} d H(s) = 1 \,. 
\end{aligned}
\end{equation}
By \eqref{eq:uTilde2} and \eqref{eq:GabStrict} the conditional CDF $G_{[\hat{a},\hat{b}]}$ is feasible in the  problem  above.
We claim that it is also optimal.
Suppose, towards a contradiction, that there exist a CDF $H$ that is feasible in \eqref{eq:linear-problem-no-feasibility} and achieves a strictly higher value than $G_{[\hat{a},\hat{b}]}$.  Consider the CDF
\[
    K(z) = \begin{cases}
    G(z) &\text{ if } z \in [0,1] \setminus [\hat{a},\hat{b}] \\
    G(\hat a_-) + H(z) (G(\hat{b})-G(\hat{a}_-)) &\text{ if } z \in [{\hat{a},\hat{b}}],
    \end{cases}
\]
which equals $G$ outside the interval $[\hat{a},\hat{b}]$ and $H$ conditional on being in $[\hat{a},\hat{b}]$.
Using \eqref{eq:uTilde}, the
definition of $M(z)$, and the feasibility of $H$ in \eqref{eq:linear-problem-no-feasibility},   it can be readily verified that this CDF is feasible in the original problem \eqref{eq:linear-problem}. Moreover, it  achieves a higher value than $G$,
 since $H$ achieves strictly higher value than $G_{[\hat{a},\hat{b}]}$ in
 \eqref{eq:linear-problem-no-feasibility}. However, this leads  to a contradiction to the optimality of $G$ in \eqref{eq:linear-problem}, thereby implying that 
 $G_{[\hat{a},\hat{b}]}$ is optimal in \eqref{eq:linear-problem-no-feasibility}.


Next, we establish that  there cannot exist an optimal solution $H$ to the problem \eqref{eq:linear-problem-no-feasibility} where for some $z \in (\hat{a},\hat{b})$ 
\begin{equation} \label{eq:bindingProperty}
    \int_z^{\hat b} H(s) ds = M(z).
\end{equation}
Suppose such an optimal solution exists. Then,  $K$ would be an optimal solution to the original problem satisfying $z \in B_K \supset B_G$,
where
$B_K$ is defined as in   \eqref{eq:mpc-proof-main-result} (after replacing $G$ with $K$) and  is the set of points where the mean preserving contraction constraint binds.
However, this contradicts that
$G$ is a  solution to the original problem that is maximal (in terms of the set where the MPC constraints bind). 

We next consider a relaxed version of the optimization problem \eqref{eq:linear-problem-no-feasibility} where 
we replace the second constraint of  \eqref{eq:linear-problem-no-feasibility} with a constraint that ensures that $H$
has the same mean as~$G_{[\hat a,\hat b]}$:
\begin{equation*}
\begin{aligned}
        \max_{H}& \int_\Omega v(s) d H(s)  \\
        \text{ subject to }& \int_\Omega \tilde{u}_i(s) d H(s) \geq 0 &&\text{ for } i \in \{1,\ldots,n\}\\
         & \int_{\hat{a}}^{\hat{b}}  H(s) ds = \int_{\hat{a}}^{\hat{b}}  G_{[\hat{a},\hat{b}]}(s) ds \\
         & \int_{[\hat{a},\hat{b}]} d H(s) = 1 \,. 
\end{aligned}
\end{equation*}
By Lemma~\ref{lem:local-problem} there exists a solution $J$ to this relaxed problem that is the sum of $n+2$ mass points.
Since $G_{[\hat a,\hat b]}$ is feasible in this problem, it readily follows that
\begin{equation} 
\label{eq:JSolutionBetter}
    \int_\Omega v(s) dJ(s) \geq \int_\Omega v(s) dG_{[\hat a,\hat b]}(s).
\end{equation}
Suppose, towards a contradiction,  that there exists  $z \in [\hat{a},\hat{b}]$ such that 
\begin{equation} \label{eq:jConstraint}
    \int_z^{\hat{b}} J(s) ds < M(z) \,.
\end{equation}
Then, by Lemma~\ref{lem:construct-feasible-sol}, there exists some $\lambda \in (0,1)$ such that $(1-\lambda) G_{[\hat a,\hat b]} + \lambda J$ satisfies 
\begin{equation}\label{eq:mpc-aux}
    \int_r^{\hat{b}} (1-\lambda) G_{[\hat a,\hat b]} (s) + \lambda J(s) ds \geq M(r) \,, 
\end{equation}
for all $r \in [\hat a,\hat b]$, and the inequality holds
with equality for some $r \in [\hat a,\hat b]$.
This implies that
  $(1-\lambda) G_{[\hat a,\hat b]} + \lambda J$ is feasible 
   for the problem \eqref{eq:linear-problem-no-feasibility}. Furthermore, 
   by the linearity of the objective,  \eqref{eq:JSolutionBetter},
  and the optimality of $G_{[\hat a,\hat b]}$ in \eqref{eq:linear-problem-no-feasibility}, 
  it follows that  $(1-\lambda) G_{[\hat a,\hat b]} + \lambda J$ is also optimal in  \eqref{eq:linear-problem-no-feasibility}. 
  However, this leads to a contradiction to the fact that  \eqref{eq:linear-problem-no-feasibility} does not admit an optimal solution where the equality in  
\eqref{eq:bindingProperty} holds for some $z\in [\hat a,  \hat b] \subset [a,b]$.

Consequently, the inequality \eqref{eq:jConstraint} cannot hold, and 
$J$ must be feasible in problem \eqref{eq:linear-problem-no-feasibility}.
Together with \eqref{eq:JSolutionBetter} this implies that $J$ is an optimal solution to \eqref{eq:linear-problem-no-feasibility} 
that assigns mass to only $n+2$ points in the interval $[\hat{a},\hat{b}]$.
This implies that the CDF
\begin{equation} \label{eq:constructNewOpt}
    \begin{cases}
    G(z) &\text{ if } z \in [0,1] \setminus [\hat{a},\hat{b}] \\
    G(\hat{a}_-) + J(z) (G(\hat{b})-G(\hat{a}_-)) &\text{ if } z \in [\hat{a},\hat{b}]
    \end{cases}
\end{equation}
is a solution of the original problem that assigns mass to only $n+2$ points in the interval $[\hat{a},\hat{b}]$.
By setting $\hat{a} = a + \frac{1}{r}$ and $\hat{b} = b - \frac{1}{r}$ we can thus find a sequence of solutions $(H^r)$ to \eqref{eq:linear-problem} that each have at most $n+2$ mass points in the interval $[a + \frac{1}{r}, b - \frac{1}{r}]$. 
As the set of feasible distributions is closed and the objective function is upper semicontinuous this sequence admits a limit point $H^\infty$ which itself is optimal in \eqref{eq:linear-problem}.
This limit distribution consists of at most $n+2$ mass points in the interval $(a,b)$.
Furthermore, by definition of $a,b$ and our construction in \eqref{eq:constructNewOpt}
each solution $H^r$ and hence $H^\infty$ satisfies the MPC constraint with equality at $\{a,b\}$. 
Thus, Lemma \ref{lem:mass-preservation} implies that $H^\infty$ is continuous at these points, and  $H^\infty(a)=F(a)$ and $H^\infty(b)=F(b)$.

Hence, we have established that for every solution $G$ for which $B_G$ is maximal, either $x \in B_G$ which by Lemma~\ref{lem:mass-preservation} implies that $G(x)=F(x)$. Or $x \notin B_G$ and then one can find a new  solution $\tilde{G}$ such that (i) $\tilde{G}$ has at most $n+2$ mass points in the interval $(a,b)$ with $a = \max\{ z \leq x \colon z \in B_G \}$ and $b = \min\{ z \geq x \colon z \in B_G \}$,
  (ii)  $\tilde{G}(a)=F(a)$ and $\tilde{G}(b)=F(b)$ which implies that the mass inside the interval $[a,b]$ is preserved, and (iii) $\tilde{G}$ matches $G$ outside $(a,b)$.
Since every interval contains a rational number there can be at most countably many such intervals. Proceeding inductively, the claim follows.
%
\end{proof}

To establish Proposition~\ref{prop:laminar}, we make use of the partition lemma, stated next:
	\begin{lemma}[Partition Lemma]\label{lem:construct}
		Suppose that distributions $F,  G$ are such that
		$\int_x^1 G(t) dt \geq \int_x^1 F(t)dt$ for $x\in I=[a,b]$, where the inequality holds with equality only for the end points of $I$.
		Suppose further that 
		$G(a)=F(a)$,
		$G(x) = G(a) + \sum_{r=1}^{K} p_r \mathbf{1}_{x\leq m_{r}}$ 
		for $x\in I$
		where $\sum_{r=1}^{K} p_r = F(b)-F(a)$,
		$(m_r)$ is a (weakly) increasing in $r$,
		and  $m_r\in I$ for $r\in [K]\equiv\{1,\dots,K\}$.

		There exists a collection of intervals
		  $\{J_r\}_{r\in [K]}$ such that 
		$\{P_k\} = \{J_k \setminus \cup_{\ell \in \mathcal{A}| \ell > k  } J_{\ell}\}$ is a laminar partition, which satisfies:
		\begin{itemize}
			\item[(a)] 
			$J_{1} = I$, and 
			if $K>1$, then
			$F(\inf J_1)<F(\inf J_K) < F(\sup J_K)<F(\sup J_1)$;
			\item[(b)] $\int_{P_k}  dF(x)=p_k$  for all ${k\in[K]}$;
			\item[(c)] $\int_{P_k} x dF(x) =  p_k m_k$ for all ${k\in [K]}$.
		\end{itemize}
		
	\end{lemma}

\begin{proof}[Proof of Lemma \ref{lem:construct}]
We prove the claim by induction on $K$. 
Note that when $K=1$ we have $J_1={P}_1=I$,
which readily implies properties (a) and (b).
In addition,  the
definition of $p_1, m_1$
implies that
\begin{equation}
    \begin{aligned}
        G(b) b -G(a)a -p_1 m_1 &= G(a) (b-a) + p_1 (b-m_1) = \int_a^b G(t) dt =\int_a^b F(t) dt \\ 
        &= F(b)b-F(a)a-\int_{I} t dF(t) = G(b) b -G(a)a- \int_{{P}_1} t dF(t).        
    \end{aligned}
\end{equation}
    Hence, property (c) also follows.
    
We proceed by considering two cases: $K=2$, $K>2$.

\underline{$K=2$}:
Let $t_1,t_2\in I$ be such that
$F(t_1)-F(a)=F(b)-F(t_2)=p_1$.
Observe that since 		$\int_x^1 G(t) dt \geq \int_x^1 F(t)dt$  $x\in I$ and this inequality holds with equality only at the end points of $I$, 
we have (i) $\int_{a}^{t_1} F(x) dx > \int_{a}^{t_1} G(x) dx $ and  (ii)
$\int_{t_2}^{b} F(x) dx < \int_{t_2}^{b} G(x) dx $. Using the first inequality and the definition of $G$ we obtain:
\begin{equation}
    \begin{aligned}
p_1&(t_1-m_1)^+ +
    G(a) (t_1-a)\leq  \int_{a}^{t_1} G(x) dx<
    \int_{a}^{t_1} F(x) dx\\
    &=
    F(t_1)t_1-F(a)a-\int_{a}^{t_1} x dF(x) 
    = (G(a) +p_1)t_1 - G(a) a  -\int_{a}^{t_1} x dF(x).
    \end{aligned}
    \end{equation}
    Rearranging the terms, this yields 
    \begin{equation} \label{eq:leftSide}
     p_1 m_1 \geq p_1 t_1 - p_1(t_1-m_1)^+ > \int_{a}^{t_1} x dF(x).
    \end{equation}
Similarly, using (ii) and the definition of $G$ we obtain:
\begin{equation}
    \begin{aligned}
    G(b)& (b-t_2)  - p_1(m_1 - t_2)^+
    \geq \int_{t_2}^{b} G(x) dx>
    \int_{t_2}^{b} F(x) dx\\
    &=
    F(b)b-F(t_2)t_2-\int_{t_2}^{b} x dF(x) 
    = G(b) b - (G(b) - p_1) t_2-\int_{t_2}^{b} x dF(x).
    \end{aligned}
\end{equation}
Rearranging the terms, this yields 
\begin{equation} \label{eq:rightSide}
    p_1 m_1 \leq p_1 t_2
    + p_1(m_1 - t_2)^+ < \int^{b}_{t_2} x dF(x).
\end{equation}
Combining \eqref{eq:leftSide} and \eqref{eq:rightSide}, and the fact that 
$F(t_1)-F(a)=F(b)-F(t_2)=p_1$
implies that
there exist $\hat{t}_1,\hat{t}_2 \in \mathrm{int}(I)$ 
satisfying 
$F(a)<F(\hat t_1)<F(\hat t_2)<F(b)$
such that
$F(\hat{t}_1) - F(a) + F(b) - F(\hat{t}_2) = p_1$ and
\begin{equation} \label{eq:condExp1}
	 \int_{a}^{\hat t_1} x dF(x) + \int_{\hat t_2}^b x dF(x) = p_1 m_1.
\end{equation}
Note that
\begin{equation*}
    \begin{aligned}
(b-a)G(a)& +  (b-m_1)p_1 + (b-m_2) p_2
= \int_{a}^{b} G(x) dx
 = \int_{a}^{b} F(x) dx \\
 &= bF(b) -aF(a)- \int_{a}^{b} x dF(x)    
 = bG(b) -aG(a)- \int_{a}^{b} x dF(x).
    \end{aligned}
\end{equation*}
Since 
$p_1+p_2 = G(b)-G(a)$, this in turn implies that $\int_{a}^{b} x dF(x)= p_1 m_1 + p_2 m_2.$
Combining this observation with \eqref{eq:condExp1}, we conclude that
\begin{equation}\label{eq:condExp2}
    \begin{aligned}
    \int_{\hat{t}_1}^{\hat{t}_2} 
    x dF(x)= p_2 m_2.
    \end{aligned}
\end{equation}
Let $J_2=[\hat{t}_1,\hat{t}_2]$, and $J_1=I$, and define $P_1, P_2$ as in the statement of the lemma.
Observe that this construction immediately satisfies (a) and (b). Moreover, (c) also follows from \eqref{eq:condExp1} and \eqref{eq:condExp2}.
Thus, the claim holds when $K=2$.

\underline{$K>2$}:
Suppose that $K>2$, and that the induction hypothesis holds for any $K'\leq K-1$.
Let $\hat{p}_2 = p_K$,   $\hat{m}_2= m_K $;
and
$\hat{p}_1= \sum_{k\in [K-1]} p_k$,  
$\hat{m}_1= \frac{1}{\hat{p}_1} \sum_{k\in[K-1]} p_k m_k $.
Define a distribution $\hat{G}$ such that
$\hat{G}(x)=G(x)$ for $x\notin I$,
$\hat{G}(a) = F(a)$, and $\hat{G}(x) = \hat{G}(a) + \sum_{r=1}^{2} \hat p_r \mathbf{1}_{x\leq \hat m_{r}}$.
This construction ensures that $\hat{p}_1+\hat{p}_2 =F(b) - F(a)$ and $\hat{m}_2 > \hat{m}_1$.
Moreover, $\hat G$ is a mean preserving contraction of ${G}$, and hence $\int_x^1 \hat{G}(t) dt \geq 
\int_x^1 {G}(t) dt$.
Since $\hat{G}(x)=G(x)$ for $x\notin I$, 
this in turn implies that $\int_x^1 \hat{G}(t) dt \geq \int_x^1 F(t) dt$ for $x\in I$ where the inequality holds with equality only for the end points of $I$.
Thus, the assumptions of the lemma hold for $\hat G$ and $F$, and using the induction hypothesis for $K'=2$,
we conclude that there exists intervals $\hat{J}_1$, $\hat{J}_2$ and sets $P_2 = \hat{J}_2$,
$P_1 = \hat{J}_1 \setminus \hat{J}_2$,
such that 
\begin{itemize}
	\item[($\hat a$)]
	$ I=\hat{J}_1 \supset \hat{J}_2$, and
$F(\inf \hat J_1)<F(\inf \hat J_2) < F(\sup \hat J_2)<F(\sup \hat J_1)$;
				\item[($\hat b$)] $\int_{P_k}  dF(x)= \hat p_k$  for  ${k\in \{1,2\}}$;
			\item[($\hat c$)] $\int_{P_k} x dF(x) =  \hat p_k \hat m_k$ for all ${k\in \{1,2\}}$.
\end{itemize}
Note that $(\hat b)$ and $(\hat c)$ imply that
$\hat{m}_2 \in \hat{J}_2$.

Denote by ${x}_0, {x}_1$ the end points of $\hat{J}_2$
and let $q_0=F(x_0)>F(a)$, $q_1=F(x_1)<F(b)$.
Define a cumulative distribution function $ F'(\cdot)$, such that 
\begin{equation} \label{eq:newCdf}
	{F}'(x)= 
	\begin{cases}
		F(x)/(1-\hat p_2) & \mbox{for $x\leq x_0$}, \\
		F(x_0)/(1-\hat p_2) & \mbox{for $x_0< x <x_1$}, \\
		(F(x) - \hat p_2) /(1 - \hat p_2) & \mbox{for $x_1\leq x$}. \\
	\end{cases}
\end{equation}
Set $p'_k= p_k /(1-\hat{p}_2)$ and ${m}'_k = m_k$ for $k\in [K-1]$. 
Let distribution $G'$ be such that
$G' (x) = G(x) / (1-\hat{p}_2) $ for $x\notin I$,
and $G'(x)= {G}'(a) + \sum_{r\in [K-1]}  p'_r \mathbf{1}_{x\leq  m'_{r}}$.
Observe that by construction ${G}'(a) = {F}'(a)$,
$\sum_{r\in [K-1]} p'_r = {F}'(b)-{F}'(a)$, and $\{m_r'\}$ is  increasing in $r$, where $m_r'\in I$,
$m_{r}'\leq\hat{m}_2$
for $r\in [K-1]$.
The following lemma implies that $G'$ and $F'$ also satisfy the MPC constraints over $I$:
\begin{lemma}
 $\int_{x}^1 G'(t) dt \geq \int_x^1 F'(t) dt$ for $x\in I$, where the inequality holds with equality only for the end points of $I$.
\end{lemma}
\begin{proof}
The definition of $G'$ implies that it can alternatively be expressed as follows:
\begin{equation} \label{eq:newCdfMu}
	{G}'(x)= 
	\begin{cases}
		G(x)/(1-\hat p_2) & \mbox{for $x< \hat {m}_{2}$}, \\
		(G(x) - \hat p_2) /(1 - \hat p_2) & \mbox{for $x\geq \hat {m}_{2}$}. \\
	\end{cases}
\end{equation}
Since $\int_{b}^1 G(t) dt  = \int_b^1 F(t) dt$,
 \eqref{eq:newCdf} and \eqref{eq:newCdfMu} 
readily imply
that $\int_{b}^1 G'(t) dt  = \int_b^1 F'(t) dt$.
Similarly, using these observations and \eqref{eq:newCdf} we have
\begin{equation} \label{eq:FPrimeEq}
    \begin{aligned}
        (1-\hat{p}_2) &\int_a^1 F'(t)  dt= \int_a^1 F(t) dt -\int_{x_0}^{x_1} F(t) dt + F(x_0) (x_1-x_0) - \hat{p}_2 (1-x_1) \\ 
        &= \int_a^1 F(t) dt -
        F(x_1) x_1 +F(x_0)x_0 + \hat p_2 \hat{m}_2
         + F(x_0) (x_1-x_0) - \hat{p}_2 (1-x_1) \\
         &= \int_a^1 G(t) dt - \hat p_2 (1-\hat{m}_2)
    \end{aligned}
\end{equation}
Here, the second line rewrites $\int_{x_0}^{x_1} F(t) dt$ 
using integration by parts,  and leverages ($\hat c$).
The third line uses the fact that $\hat p_2 = F(x_1)-F(x_0)$
and $\int_{a}^1 G(t) dt  = \int_a^1 F(t) dt$.
On the other hand,  \eqref{eq:newCdfMu}  readily implies that:
\begin{equation}
    \begin{aligned}
        (1-\hat{p}_2) \int_a^1 G'(t)  dt&= \int_a^1 G(t) dt - \hat{p}_2 (1-\hat{m}_2) \\ 
    \end{aligned}
\end{equation}
Together with \eqref{eq:FPrimeEq}, this equation implies that  $\int_{a}^1 G'(t) dt  = \int_a^1 F'(t) dt$.
Thus, the inequality in the claim holds with equality for the end points of $I$.

Recall that
$\hat{m}_2 \in \hat{I}_2 $ and hence
$a<x_0 \leq  \hat{m}_2 = m_K  \leq x_1 <b$. We complete the proof by focusing on the value $x$ takes in the following cases: (i) $a< x \leq x_0$, (ii) $x_0 \leq x \leq \hat{m}_2  $, (iii) $\hat{m}_2 \leq x \leq x_1$,
(iv) $x_1 \leq x < b$.

\noindent{\bf Case (i):}
Using the observations
$
 \int_x^1 G(t) dt   >   \int_x^1 F(t) dt 
$
and $\int_a^1 G(t) dt   =   \int_a^1 F(t) dt $
together with   \eqref{eq:newCdf} and \eqref{eq:newCdfMu} yields
\begin{equation}
    \int_a^x G'(t) dt  = \frac{1}{1-\hat{p}_2} \int_a^x G(t) dt   < \frac{1}{1-\hat{p}_2} \int_a^x F(t) dt = \int_a^x F'(t) dt. 
\end{equation}
Together with $\int_{a}^1 G'(t) dt  = \int_a^1 F'(t) dt$ this implies that
$\int_x^1 G'(t) dt >\int_x^1 F'(t) dt  $ in case (i).

\noindent \textbf{Case (ii):} 
Using   \eqref{eq:newCdf} and \eqref{eq:newCdfMu}  we obtain:
\begin{equation*}
\begin{aligned}
({1-\hat{p}_2})\!\!\int_x^{1} G'(t)-F'(t) dt  &= \!\!\int_x^{1} \!\!G(t) dt - (1-\hat m_2) \hat{p}_2 \! -\! \int_{x_1}^1 F(t)dt- \int_x^{x_1} F(x_0) dt+ (1-x_1) \hat{p}_2 \, . \\
\end{aligned}
\end{equation*}
Since $G$ is an increasing function, it can be seen that the right hand side is a concave function of $x$. Thus, for $x\in [x_0,\hat{m}_2]$ this expression is minimized for $x=x_0$ or $x=\hat{m}_2$.
For $x=x_0$, case (i) implies that the expression is non-negative.  We next argue that for $x=\hat{m}_2$ the expression remains non-negative. This in turn implies that $\int_x^{1} G'(t)-F'(t) dt  \geq 0$ for $x\in [x_0,\hat{m}_2]$, as claimed.

Setting $x=\hat{m}_2$, 
recalling that
$\int_{b}^1 G(t) dt  = \int_b^1 F(t) dt$, 
and observing that $G(t) =G(b)=F(b)$ for $t\in [\hat{m}_2,b]$
the right hand side of the previous equation reduces to:
\begin{equation}
\begin{aligned}
R:&= (b-\hat m_2)F(b) -(1-\hat{m}_2)  \hat{p}_2  -\int_{x_1}^b F(t) dt - (x_1-\hat{m}_2) F(x_0) +(1-x_1) \hat{p}_2\\
&=(b-\hat m_2)F(b)    -\int_{x_1}^b F(t) dt - (x_1-\hat{m}_2) F(x_0) -(x_1-\hat{m}_2) \hat{p}_2\\
&=(b-x_1)F(b)-\int_{x_1}^b F(t) dt  + (x_1-\hat m_2)(F(b)-  F(x_0) -  \hat{p}_2).
\end{aligned}
\end{equation}
Since $F(b)\geq F(x_1) = \hat{p}_2 + F(x_0)$, we conclude:
\begin{equation}
    \begin{aligned}
        R &\geq  
        (b-x_1)F(b)    -\int_{x_1}^b F(t) dt \geq 0,
    \end{aligned}
\end{equation}
where the last inequality applies since $F$ is weakly increasing.
Thus, we conclude that $\int_{\hat{m}_2}^{1} G'(t)-F'(t) dt \geq 0$, and the claim follows.

\noindent{\bf Case (iii):} First observe that \eqref{eq:newCdf} and \eqref{eq:newCdfMu} imply that
\begin{equation*}
\begin{aligned}
({1-\hat{p}_2})\int_x^{1} G'(t)-F'(t) dt  &= \int_x^{1} G(t) dt - (1-x) \hat{p}_2  - \int_{x_1}^1 F(t)dt- \int_x^{x_1} F(x_0) dt+ (1-x_1) \hat{p}_2 . \\
\end{aligned}
\end{equation*}
Similar to case (ii), the right hand side is a concave function of $x$. Thus,  for $x\in [\hat{m}_2,x_1]$ this expression is minimized for $x=\hat{m}_2$ or $x={x}_1$.
When $x=\hat{m}_2$, case (ii) implies that $\int_x^{1} G'(t)-F'(t) dt \geq 0$.
Similarly, when $x=x_1$, case (iv) implies that $\int_x^{1} G'(t)-F'(t) dt \geq 0$. Thus, it follows that  $\int_x^{1} G'(t)-F'(t) dt  \geq 0 $  for all $x\in [\hat{m}_2,x_1]$.

\noindent{\bf Case (iv):} In this case, \eqref{eq:newCdf} and \eqref{eq:newCdfMu} readily imply that
\begin{equation*}
\begin{aligned}
({1-\hat{p}_2})\int_x^{1} G'(t)-F'(t) dt  &=  \int_x^{1} G(t)-F(t) dt >0,
\end{aligned}
\end{equation*}
where the inequality follows from our assumptions on $F$ and $G$. 
\end{proof}
 
Summarizing, we have established that
the distribution $G'$ and $F'$ satisfy the conditions of the lemma.  By the induction hypothesis, we have that 
there exists  intervals   $\{J_k'\}_{k\in  [K-1]}$ 
and sets
$P_k' = J_k' \setminus \cup_{ \ell \in [K-1] | \ell>k } J_{\ell}'$ 
for all $k\in \mathcal{A}'$ such that:
\begin{itemize}
		\item[(a')] 
			$J_{1}' = I$, and 
$F(\inf  J_1')<F(\inf  J_{K-1}') < F(\sup J_{K-1}')<F(\sup  J_1')$;
			\item[(b')] $\int_{P_k'}  dF'(x)=p_k'$  for all ${k\in[K-1]}$;
			\item[(c')] $\int_{P_k'} x dF'(x) =  p_k' m_k'$ for all ${k\in [K-1]}$.
\end{itemize}


Let $J_k=J'_k \setminus \hat J_2$ for $k\in [K-1]$
such that $\hat J_2 \not\subseteq J_k'$,  and
$J_k=J'_k$ for the remaining $k\in [K-1]$.
Define $J_K=\hat{J}_2 = [x_0,x_1]$.
For $k\in [K]$, 
let
$P_k = J_k \setminus \cup_{\ell \in [K] |\ell>k } J_{\ell}$.
Note that
the definition of 
the collection $\{P_k\}_{k\in [K]}$ implies that it  constitutes a laminar partition of $I$.
Observe that the construction of $\{J_k\}_{k\in {[K]}}$ and ($\hat a$), ($a'$) imply that
these intervals also satisfy condition (a) of the lemma.
Note that by construction we have
\begin{equation}\label{eq:TkProperties}
	P_k \subseteq  P_k' \subseteq P_k \cup J_K \quad \mbox{and} \quad   P_k \cap J_K=\emptyset \quad \mbox{ for $k\in [K-1]$.}
\end{equation}
Since 
$\int_{J_K} dF'(t) = 0$ by \eqref{eq:newCdf}
this observation implies that
$\int_{P_k'} dF'(t) =\int_{P_k} dF'(t)$ for $k\in [K-1]$.

Using  \eqref{eq:newCdf},
($b'$), 
and \eqref{eq:TkProperties}, this observation 
implies that
\[
\int_{P_k} dF(t)=
\int_{P_k} dF'(t) (1-\hat{p}_2) 
=
\int_{P_k'} dF'(t) (1-\hat{p}_2) = p_k' (1-\hat{p}_2)= p_k,
\]
for $k\in [K-1]$.
Similarly, by ($\hat b$) we have
$\int_{P_K} dF(t)=  \int_{\hat P_2} dF(t)= \hat{p}_2 = p_K $.


Finally, observe that by ($\hat c$) we have
$\int_{ P_{K}} t  dF(t)= 
\int_{ \hat P_2} t dF(t)= \hat{p}_2 \hat{m}_2 = p_K m_K $.
Similarly, 
\eqref{eq:newCdf}  and \eqref{eq:TkProperties}
imply that for $k\in[K-1]$, we have
$$\int_{ P_{k}} t  dF(t)= (1-\hat{p}_2) \int_{ P_{k}} t  dF'(t)= (1-\hat{p}_2) \int_{ P_{k}'} t  dF'(t)=  (1-\hat{p}_2) p_k' m_k' = p_k m_k.$$

These observations imply that the constructed $\{J_k\}_{k\in [K]}$ and $\{P_k\}_{k\in [K]}$ satisfy the induction hypotheses (a)--(c) for $K$. Thus, the claim follows by induction.
\end{proof}

\begin{proof}[Proof of Proposition \ref{prop:laminar}]
	By definition, the interval $I_j$  in the statement of Proposition \ref{prop:laminar} satisfies the conditions of Lemma \ref{lem:construct},
	(after setting $a=a_j$, $b=b_j$).  The lemma 
	defines auxiliary intervals $\{J_r\}$ and 
	explicitly constructs a laminar partition that satisfies  conditions (a)-(c). Here, conditions (b) and (c) readily imply that the constructed laminar partition satisfies the claim in Proposition \ref{prop:laminar}, concluding the proof.
\end{proof}	

%
%
%
%
%
\begin{proof}[Proof of Theorem~\ref{thm:optimal-signals}]

The existence of an optimal mechanism follows from standard compactness arguments and is proven in Lemma~\ref{lem:existence} in the Online Appendix.
Consider now an arbitrary optimal SGR mechanism.
Fix a type profile $\theta$.
By combining Lemma~\ref{lem:decoupled-n1} and \ref{lem:decoupled} we can replace $G^\theta$ by another solution to the respective optimization problem under MPC and linear side constraints and obtain a new SGR mechansim.
By Propositions~\ref{prop:general-result} and \ref{prop:laminar} there always exists a solution to this optimization problem under MPC and linear side constraints that can be implemented by a laminar partitional signal.
Iterating this process over type profiles we get that there exists an optimal SGR mechanism in which each distribution $G^\theta$ can be implemented by a laminar partitional signal.
Thus, we constructed an optimal laminar partitional mechanism.
\end{proof}

\begin{proof}[Proof of Proposition \ref{prop:Thm2ndPart}] Part (i) of the proposition follows from Proposition~\ref{prop:laminar} by counting the number of linear side constraints in the optimization problem stated in Lemma~\ref{lem:decoupled-n1}. 
Similarly, by counting the side constraints in the optimization problem stated in Lemma~\ref{lem:decoupled} it follows that for every optimal SGR mechanism $(G^\ast,q^\ast)$ and every type profile $\theta$ there exists a laminar partitional signal that generates $G^{\ast,\theta}$.
This proves the result if the  action profile distribution 
 $q^{\ast,\theta}(\cdot|m)$
is degenerate
and deterministically recommends an
  action profile for each posterior mean $m$ and type profile $\theta$.
  We refer to  SGR mechanisms associated with such action profile distributions as non-randomized SGR mechanisms (since for a given type profile and posterior mean, their recommendation is deterministic).
In the multi-agent case, in contrast to the single agent case, 
non-randomized SGR mechanisms need not be optimal
and the designer may need to use
non-degenerate
distributions of recommended action profiles.

To show the result 
when optimal SGR mechanisms require
non-degenerate action profile distributions, we  consider the parametrization of SGR mechanisms in terms of the mean $m_{a,\theta}$ and the \emph{unconditional} probability of each action profile $q^\theta(a)$ (see Section~\ref{se:sgrm}).
In this parametrization the set of non-randomized SGR mechanisms is dense in the space of all SGR mechanisms (as  for each randomized mechanism an arbitrarily small perturbation of all the means
induces a  non-randomized mechanism).
The designer's payoff as well as the IC constraints are continuous in this parametrization (see the proof of Lemma~\ref{lem:existence}), and hence  there exists a sequence of non-randomized SGR mechanisms such that the payoff of the designer converges to the value of the optimal (randomized) SGR mechanism along the sequence.
By the argument of Lemma~\ref{se:decoupling} for each of these non-randomized SGR mechanisms there exists a laminar partitional mechanism 
with weakly larger payoff and  partition  depth of  at most $\sum_{i \in N} |A_i|^2 |\Theta_i|+2$. Since there are finitely many partial orders defining laminar partitions, there exist one that appears infinitely often along the sequence.
Since for a given partial order the laminar partition is defined in terms of the end points of the convex hulls of the partition elements which belong to $[0,1]^{|A|}$, there is a subsequence associated with this partial order that converges to a laminar partition consistent with the same partial order -- which still has depth bounded by $\sum_{i \in N} |A_i|^2 |\Theta_i|+2$.
Moreover, the designer's payoff is continuous in the end points of the aforementioned intervals  (since the distribution of the state is continuous).
Thus, this limit point defines a new laminar partitional mechanism which achieves the optimal objective and has at most the depth stated in the claim.
\end{proof}

\begin{proof}[Proof of Proposition \ref{prop:public-bound}]
The first claim is immediate and follows as explained in the text. Here, we focus on the following example and use it to prove the second part of the claim.

\begin{example}\label{ex:public-private}
There is a single agent,  all types are equally likely, i.e., $\phi(\theta) \equiv 1/|\Theta|$ for all $\theta\in \Theta=\{1,\dots,n\}$, and the state is uniformly distributed in $[0,1]$.
For $k \in \{ -2n , \ldots , 2n \}$ we define intervals $B_{L,k} = [b_{L,k-1},b_{L,k}], B_{R,k} = [b_{R,k-1},b_{R,k}]$.
Here, for any integer $k$ we let
\begin{equation*}
    b_{L,k} = \frac{1}{4} + \frac{1}{8} \text{sgn}(k) \sqrt{\frac{|k|}{2 n }} \hspace{2cm}
    b_{R,k} = \frac{3}{4} + \frac{1}{8} \text{sgn}(k) \sqrt{\frac{|k|}{2 n }} \,.
\end{equation*}
All types of the agent share the same indirect utility function $\bar{u}$, such that $\bar{u}(m,\theta) = (m-\frac{1}{2})^2$ for all $m \in \{b_{L,k},b_{R,k} \}$, and linearly interpolated in each $B_{L,k}$ and $B_{R,k}$ (in our construction the payoffs outside these intervals will be immaterial).
The indirect utility of the designer is 
\begin{equation} \label{eq:constructV}
    \bar{v}(m,\theta) = \begin{cases}
    1 &\text{ if }  m \in B_{L,2\theta} \cup B_{L,- 2\theta } \cup B_{R,2 n+ 2- 2\theta} \cup B_{R, -2n -2+2 \theta } \\ 
    0 &\text{ otherwise.} 
    \end{cases} 
\end{equation}
The agent's indirect utility functions can be generated by taking the set of actions to be $\{a_{L,k},a_{R,k}\}$ for $k \in \{-2n-1,\ldots,+2n+1\}$ and the utilities as a function of the action to be
\[
    {u}(a_{\cdot,k}, \omega , \theta  ) = c_{\cdot ,k-1}^2 + \frac{\omega - b_{\cdot,k-1}}{b_{\cdot,k} - b_{\cdot,k-1}} (c_{\cdot,k}^2 - c_{\cdot,k-1}^2),
\]
 where $c_{\cdot ,k}=b_{\cdot ,k}-\frac{1}{2}$.
Similarly, we let $v( a, \omega , \theta )=1$ for  actions $a_{L,2\theta}, a_{L,-2\theta}, a_{R, 2n+2 - 2\theta}, a_{R, -2n-2 + 2\theta}$ and zero otherwise.
\end{example}

We begin by establishing that in the setting of Example~\ref{ex:public-private} no public mechanism achieves more than $1/|\Theta|$. 
Note that
by our construction in \eqref{eq:constructV},
for any posterior mean $m$ the indirect utility of the designer equals $1$ for at most a single type, i.e.,
$\sum_{\theta \in \Theta} \bar{v}(m,\theta) \leq 1$.
As $\phi(\theta) = 1/|\Theta|$ this immediately implies that for any \emph{type independent} distribution of the posterior mean $G$ the designer can  achieve a payoff of at most $1/|\Theta|$.

Next consider the following private mechanism:
The distribution $G^\theta$ for type $\theta \in \Theta$ consists of 4 equally likely mass points at 
$    b_{L,2 \theta},  b_{L,-2 \theta},  b_{R,2n + 2 - 2 \theta},  b_{R, - 2n - 2 + 2 \theta}$.
It is straightforward to see that the signal based on the partition  $(\Pi_k)_{k=1}^4$ with
$\Pi_1= (b_{L,2\theta}-1/8, b_{L,2\theta}+1/8)$, $\Pi_2 = [0,1/2] \setminus \Pi_1$, $\Pi_3=[ b_{R,2n + 2 - 2 \theta}-1/8, b_{R,2n + 2 - 2 \theta}+1/8]$, $\Pi_4=(1/2,1] \setminus \Pi_3$ induces the desired posterior mean distribution.
At each of these beliefs the agent's indirect utility is given by $\bar{u}(m,\theta) = (m-\frac{1}{2})^2$.
Thus, the benefit  the agent of type $\theta'$ derives from observing the signal meant for type $\theta$ (relative to observing no signal) equals the variance of $G^\theta$. 

Note that the variance conditional on the posterior being less than $1/2$  equals $\nicefrac{1}{2} (b_{L,2 \theta}-\nicefrac{1}{4})^2 + \nicefrac{1}{2} (b_{L,-2 \theta} - \nicefrac{1}{4} )^2 = \frac{\theta}{64 \cdot n }$ 
and the variance conditional on the posterior  being greater than $\nicefrac{1}{2}$ equals $\nicefrac{1}{2} (b_{R,2n + 2 - 2 \theta}-\nicefrac{3}{4})^2 + \nicefrac{1}{2} (b_{R,- 2n - 2 + 2 \theta} - \nicefrac{3}{4} )^2 = \frac{n + 1 - \theta}{64 \cdot n } \,.$
By the law of the total variance the variance of $G^\theta$ thus equals $\frac{1}{2} \frac{\theta}{64 \cdot n } + \frac{1}{2} \frac{n + 1 - \theta}{64 \cdot n } + \frac{1}{2} \frac{1}{4}^2 + \frac{1}{2} \frac{1}{4}^2  = \frac{9n+1}{128\cdot n}$.
Since, this quantity does not depend on $\theta$, we conclude that each type derives equal utility from any signal and the mechanism is incentive compatible.
Each mean in the support $G^\theta$ persuades the agent to take an action that yields a payoff of $1$ to the designer. Hence, this mechanism with private signals yields a payoff of~$1$.
\end{proof}

{
\setlength{\bibsep}{0pt plus 0.ex}
\bibliographystyle{aer}
\bibliography{persuasion}
}

\newpage
\title{Online Appendix}
\maketitle


\section{Optimal Mechanisms with Large Laminar Depth}  \label{app:ZSexample}

\begin{figure}[p]
\centering
\begin{subfigure}[t]{\textwidth}
  \centering
  \caption{}
 \raisebox{.8cm}{\includegraphics[width=.4\linewidth]{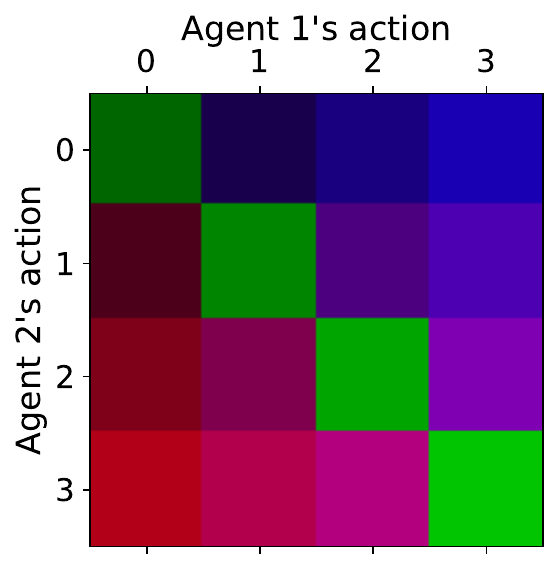}}
  
  \label{fig:sub-second}
\end{subfigure}
\newline
\vspace{.3cm}

\caption*{(b)}
\begin{tabular}{c c}
    \raisebox{.6cm}{ $\theta=(0,0)$:} & \includegraphics[width=13cm]{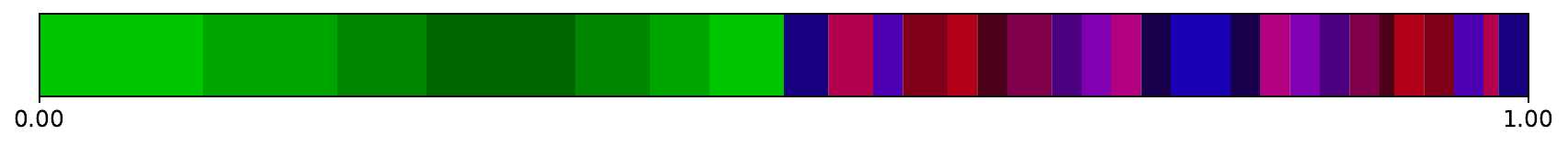}\\
\raisebox{.6cm}{ $\theta=(0,1)$:} & \includegraphics[width=13cm]{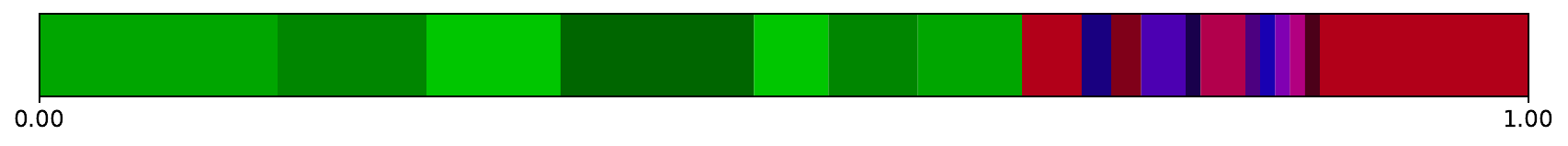}\\
\raisebox{.6cm}{ $\theta=(1,0)$:} & \includegraphics[width=13cm]{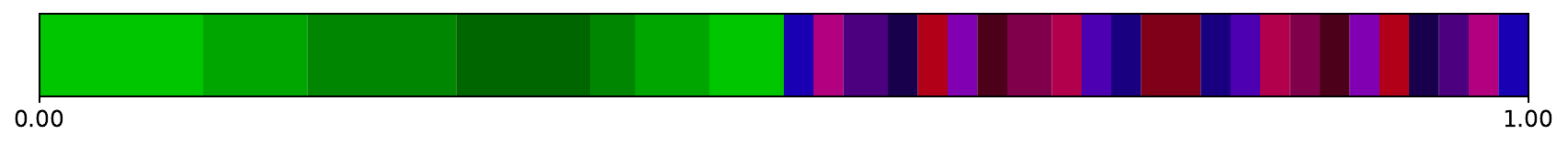}\\
\raisebox{.6cm}{ $\theta=(1,1)$:} & \includegraphics[width=13cm]{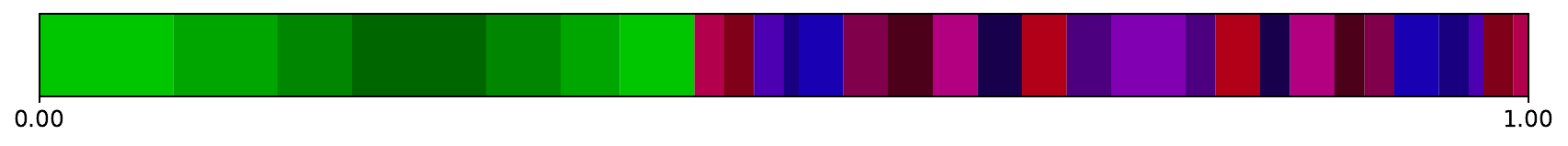}\\
\end{tabular}
    \label{fig:sub-third}
   \caption{
   Optimal mechanism.
   (a) shows the colors associated with each strategy profile.  (b) is the laminar partitional signals that constitute an optimal mechanism. Different shades of green encode the strategy profiles where the designer achieves nonzero payoff. For all type profiles, such strategy profiles are associated with smaller states. The laminar partitional signals in this example have depth 12.
    }
    \label{fig:ZS}
    \end{figure}

In this section we provide an example where all optimal laminar partitional mechanisms have depth exceeding $|\Theta| + 2$ (i.e., the depth in the single-agent case, see Proposition \ref{prop:Thm2ndPart}~(i)). 
There are two players $N=\{1,2\}$ with two possible types $\Theta_1=\Theta_2=\{0,1\}$ and four possible actions each $A_1=A_2=\{0,1,2,3\}$.
For convenience, we order type profiles
and define a function $\delta:\Theta_1\times \Theta_2\rightarrow \{0,1,2,3\}$ such that
$\delta(0,0)=0$, $\delta(0,1)=1$,
$\delta(1,0)=2$, $\delta(1,1)=3$.
The players play a zero-sum game. 
The payoff matrix of the row player for the type profile $\theta = (\theta_1,\theta_2)$ is:
$\omega(I+P_{\delta(\theta)})$, where $P_k$ is the permutation matrix whose $(\ell_1,\ell_2)$-th entry is one if $\ell_2-\ell_1=k  \mathrm{~mod(4)}$.
The state  $\omega$ is distributed uniformly on $[0,1]$.
The type profile distribution is such that $\phi(0,0)=0.1$, $\phi(0,1)=0.2$, $\phi(1,0)=0.3$, $\phi(1,1)=0.4$.
The state and the types are distributed  independently. The designer's payoff is $1$ if $a_1=a_2$ and $0$  otherwise.

An optimal mechanism is   given in Figure \ref{fig:ZS}.
As can be seen from this figure  the depth of the laminar family supporting the optimal information structure is larger than
 $|\Theta| + 2=6$.
We numerically verified that any other laminar partitional mechanism mechanism that is optimal also has depth greater than $6$.
Furthermore, 
when the number of actions is smaller (and the type space is the same) for any payoff structure laminar families of smaller depth suffice.
Conversely, when the number of actions is larger,  even with the same type space it is possible to obtain even deeper laminar families 
at the optimal mechanism
for variants of this example.

\section{A Finite-Dimensional Formulation for the Multi-Agent Case} \label{app:optFormMultiAgent}
In the single-agent case, when the agent has finitely many actions Section \ref{subse:finiteAction} established that it is possible to obtain the optimal mechanism by solving a finite-dimensional convex program. This simplification was partly driven by two factors: 
(i)  the agent can perfectly infer the posterior mean from the action recommendation
(ii) the posterior mean levels that induce a given action can be characterized explicitly given the agent's payoff function.
These factors allowed us to remove the recommended action from the problem and express it purely in terms of posterior means.
As these factors are not present in the multi-agent case, it is unclear whether one can obtain the optimal mechanism through a solution of a finite-dimensional optimization problem. We next argue that indeed through the solution of finite-dimensional programs it is possible to obtain an optimal mechanism, for as long as the agents have finitely many actions. 

Consider the formulation in \eqref{eq:optNewForm}. Note that for any given profile $\theta$ the distribution $q^\theta$ over type profiles determines the action profiles recommended at different posterior mean levels.
Since no action profile is recommended with positive probability at two different posterior mean levels, it means that action profiles are ordered according to the posterior mean levels that induce them. Denote by $\delta^\theta$ this order: $\delta^\theta(a)\geq \delta^\theta(a')$ if posterior mean that induces  $a$ is larger than that associated with $a'$ when the type profile is $\theta$.

Following an approach similar to the one in Section \ref{subse:finiteAction}, we can now express the designer's problem as follows:
\begin{equation*}
	\begin{aligned}
\max_{\{\delta^\theta\}_\theta}	&	\max_{\substack{p  \in (\Delta^{|A|})^\Theta\\z \in \RR_+^{ |A| \times |\Theta| }\\ y_i \in \RR^{  |A_i| \times |\Theta|^2 } }}  \quad  \sum_{\theta \in \Theta} \phi(\theta) \, \sum_{a \in A} p_{a,\theta} v(a,\theta) \\
		s.t.\, & \sum_{a | \delta^\theta(a) \geq \delta^\theta(\ell) } z_{a,\theta} \leq 
		\int_{1-  \sum_{a | \delta^\theta(a) \geq \delta^\theta(\ell)  }p_{a,\theta} }^{1} F^{-1}(x) dx
		&&  \forall\,\theta\in \Theta,\ell > 1 , \\
&\sum_{{a\in A} } z_{a,\theta} = \int_\Omega  F^{-1}(x) dx &&  \forall\,\theta\in \Theta, \\
		&	\sum_{\theta_{-i}} \phi(\theta) \sum_{a_{-i}}	(u_{i1}(a,\theta) z_{a,\theta}+u_{i2}(a,\theta) p_{a,\theta}) 
		\\
		&\qquad \geq \sum_{\theta_{-i}} \phi(\theta) \sum_{a_{-i}}	(u_{i1}(a_i',a_{-i},\theta) z_{a,\theta}+u_{i2}(a_i',a_{-i},\theta) p_{a,\theta}) && \forall i, \theta_i, a_i, a_i'\\ 
	& y_{i,\theta_i,\theta_i',a_i}\geq	\sum_{\theta_{-i}} \phi(\theta) \sum_{a_{-i}}	(u_{i1}(a_i',a_{-i},\theta) z_{a,\theta_i',\theta_{-i}}+u_{i2}(a_i',a_{-i},\theta) p_{a,\theta_{i}',\theta_{-i}} ) && \forall i,\theta_i,\theta_i',a_i,a_i'\\
&\sum_{a_i} \sum_{\theta_{-i}} \phi(\theta) \sum_{a_{-i}}	(u_{i1}(a,\theta) z_{a,\theta}+u_{i2}(a,\theta) p_{a,\theta})  \geq 
\sum_{a_{i}} y_{i,\theta_i,\theta_i',a_i} && \forall i, \theta_i,\theta_i'\\
& {z_{a,\theta}}{p_{a',\theta}} \geq {z_{a',\theta}}{p_{a,\theta}}
&& \forall \theta,  \delta^\theta(a) \geq \delta^\theta(a')\\
& {z_{a,\theta}}\leq {p_{a,\theta}} && \forall a,\theta.
	\end{aligned}
\end{equation*}
In this optimization problem, $p_{a,\theta}$ denotes the probability with which strategy profile $a$ is induced when the type profile is $\theta$, and
$m_{\theta,a}=z_{a,\theta}/p_{a,\theta}$ 
is the corresponding posterior mean level. 
Note that $\{p_{a,\theta},z_{a,\theta}\}_a$ tuple constitutes a reparameterization of $G^\theta$.
For a given order $\delta^\theta$ on posterior mean levels,  the first two constraints amount to a restatement of the MPC constraint  $G^\theta \succeq F$. 
Note that if agents report their types truthfully and follow the action recommendations, the payoff of agent $i$ for type profile $\theta$ and action recommendation profile $a$ is given by $u_{i1}(a_i',a_{-i},\theta) z_{a,\theta_i',\theta_{-i}}/p_{a,\theta_{i}',\theta_{-i}}+u_{i2}(a_i',a_{-i},\theta) $. This implies that his expected payoff\footnote{As explained in the main text, this quantity is actually equal to the expected payoff times $\sum_{\theta_{-i}} \phi(\theta)$. With some abuse of terminology, throughout the online appendix we ignore this normalization and refer to such quantities as payoffs.} is given as in the left hand side of the third constraint. Similarly, the right hand side is the payoff from taking action $a_i'$. Thus, the third constraint ensures that if agents report their type truthfully and agent $i$ gets the action recommendation $a_i$, any deviation reduces his payoff.
Suppose that agent $i$ is of type $\theta_i$ but he misreported his type as $\theta_i'$, and received action recommendation $a_i$. Assuming all agents still truthfully report their types and follow action recommendations, what is $i$'s payoff from taking action $a_i'$? 
The right hand side of the fourth constraint captures this quantity.  At the optimal solution, the left hand side, $y_{i,\theta_i,\theta_i',a_i}$ equals the maximization of this quantity over $a_i'$, which is the best payoff $i$ can guarantee after the 
 type report $\theta_i'$ and action recommendation $a_i$.
Aggregating these terms  over all $i$ 
yields 
the right hand side of the fifth constraint, which is
the expected payoff of agent $i$ from misreporting his type as $\theta_i'$. The left hand side is the payoff from truthful reporting and following action recommendations. Thus, the fifth constraint ensures that agent $i$ has no incentive to misreport his type. The  sixth constraint
can be equivalently written as
$m_{\theta,a}=z_{a,\theta}/p_{a,\theta} \geq m_{\theta,a'}=z_{a',\theta}/p_{a',\theta}$. This ensures
  that
the $\{p_{a,\theta},z_{a,\theta}\}_a$ tuple and the associated distribution $G^\theta$ is consistent with $\delta^\theta$ in terms of the ranking of the posterior means of strategy profiles. Finally, the last constraint (together with the nonnegativity of $p_{a,\theta},z_{a,\theta}$) ensures that the posterior means are between $0$ and $1$.

To solve this problem, we can first fix  $\delta^\theta$ in the outer problem and solve the associated inner problem. Then, we can search over the orders $\delta^\theta$ (of which there are finitely many) in the outer problem.
There are two challenges with this approach. First, the number of orders to consider in the outer problem can be large. Second,   unlike the formulation in Section \ref{subse:finiteAction}, due to the sixth constraint the inner  problem is not convex.

It turns out that it is possible to overcome both challenges. Let us start with the second challenge. Despite the fact that the inner problem is nonconvex, a locally optimal solution can be obtained using, e.g., gradient ascent. If at a locally optimal solution, the nonconvex constraints are not binding then, it follows that the solution is locally optimal in the problem where these constraints are relaxed. However, the latter problem is convex and local optimality implies global optimality. Thus, the aforementioned solution is a globally optimal solution to the  inner problem. In all our numerical experiments (including the Cournot example discussed in Section \ref{se:motivatingExample}) this was the case, i.e., when we obtained a locally optimal solution using a solver, we observed that the nonconvex constraints did not bind and verified global optimality of the said solution.

The first challenge is problem specific, but the search can be drastically reduced in some cases. For instance, observe that in the Cournot example of Section \ref{subse:finiteAction}, there are $9$ strategy profiles, and naively there are $9!$ orders to consider. However, due to the symmetry in the problem it can be readily seen that the posterior means associated with strategy profiles $(a_i,a_j)$ and $(a_j,a_i)$ are identical. Furthermore, intuitively, posterior means associated with larger aggregate supply levels will be larger. That is $m_{\theta,a}>m_{\theta',a'}$ if $a_i+a_j>a_i'+a_j'$. Once this restriction imposed, together with symmetry the number of orders to consider reduces only to two (one where strategy profiles $(0,2), (2,0)$ are associated with higher posterior mean levels than $(1,1)$, and one with lower). 
Thus,
solving the inner problem for these two orders, and picking the solution
that results  in a higher payoff delivers  the optimal mechanism.
This is in fact how we obtain the optimal mechanisms in Section \ref{se:motivatingExample} (where we also numerically verify that imposing the aforementioned condition is without loss). Notably this approach allows for constructing the optimal mechanisms without discretization of the state space.
Using the approach described here the optimal solution to the optimization problem in Section \ref{se:motivatingExample} is obtained in  $\sim 20$ ms for most 
weighted combinations of $CS$ and $FP$ (using off-the-shelf interior point methods of the Knitro solver).

\section{Additional Details for the Example in Section \ref{se:buyerExample}} \label{app:detailsForBuyerExample}

Here we revisit the example in Section \ref{se:buyerExample}. 
The indirect utility $\bar{u}(m,\theta)$ of the buyer 
in this example
is given in Figure~\ref{fig:indirect}.
When the expected quality $m$ of the good is low, all types find it optimal to purchase zero units, yielding a payoff of zero. As the expected quality improves, the purchase quantity increases.
In Figure~\ref{fig:indirect}, the curve for each type is piecewise linear, and the kink-points of each curve correspond to the posterior mean levels where the agent increases his purchase quantity.
Since the state and hence the posterior mean belong to $[0,1]$, the purchase quantity of each type is at most $2$ units, and each curve in the figure has at most two kink-points. {This is easily seen as the utility any buyer type derives from consuming the third unit of the good is bounded by $(\theta + \omega) \max\{ 5 -3 , 0\}  \leq (0.6 + 1) \cdot 2 = 3.2$ which is less than the price of $10/3$.}
\begin{figure}
        \centering
        \includegraphics[width=10cm]{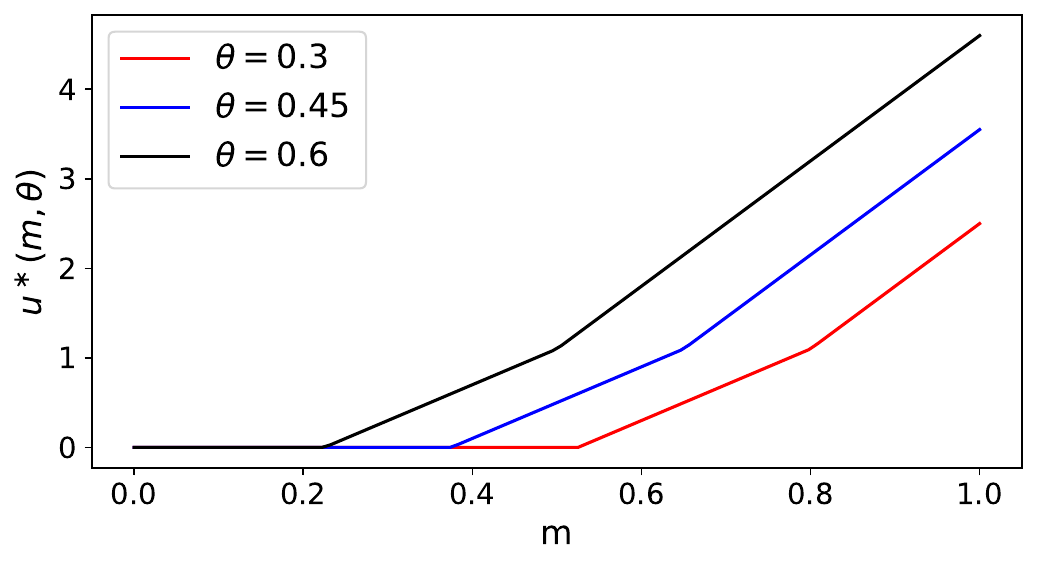}
        \caption{The indirect utility of the agent}
        \label{fig:indirect}
\end{figure}
These observations imply that in  this problem, the agent effectively considers finitely many actions, namely the quantities in $0,1$ and $2$.

The effect of the incentive compatibility constraints on the optimal mechanism are easily seen from  Figure \ref{fig:newEx}. For instance, the high type's payoff from a truthful type report is strictly positive. If this were the only relevant type, the designer could choose a strictly smaller threshold than $0.06$ and
still ensure purchase of two units
whenever state realization is above this threshold,  thereby
increasing the expected purchase amount of the high type. However, when the other types  are  also present, such a change in the signal of the high type incentivizes this type to deviate and misreport his type as low or medium. Changing the signals of the remaining types to recover incentive compatibility,  reduces the payoff the designer derives from them. The mechanism in Figure \ref{fig:newEx} maximizes the designer's payoff while carefully satisfying such incentive compatibility constraints.

As discussed in Section \ref{se:buyerExample}, in the binary action case it is without loss to focus on public  mechanisms (which do not elicit the agent's type).
In this case, one way to obtain an optimal public mechanism is to first solve for the optimal mechanism without restriction to public ones, and then reveal to each type the signals associated with all types.
By contrast, the mechanism illustrated in Figure \ref{fig:newEx} does not admit such a payoff-equivalent public implementation. For instance, under this mechanism the high type purchases two units whenever the state realization is higher than $0.06$. Suppose that this type of agent had access to the signals of, for instance, the low type as well. Then, he could infer whether the state is in $[0.06,0.16]\cup [0.94,1]$. Conditional on the state being in this set, his expectation of the state would  be approximately $0.43$. This implies that the expected payoff of the high type from purchasing the second unit is $(0.43+0.6)\times 3 -{10}/{3}<0$. Thus, for  state realizations that belong to the aforementioned set, the high type finds it optimal to strictly reduce his consumption (relative to the one in Figure \ref{fig:newEx}).
In other words, observing the additional signal reduces the expected purchase of the high type (and the other types). Hence, such a public implementation is strictly suboptimal.
As a side note, the optimal public implementation can be obtained by replacing different types with a single ``representative type'' and using the framework of Section \ref{se:analysis}.
More precisely, we can replace the designer's indirect utility with $\bar{v}(m)=\sum_\theta \phi(\theta)  \max_{a\in A(m,\theta)} v(a,m,\theta)$ 
and maximize $\int \bar{v}(m) dG(m)$ by choosing a distribution $G \succeq F$ (without any side constraints -- since  with public signals the designer does not screen the agent, and hence the IC constraints become irrelevant).
We numerically conducted this exercise and also verified that  restricting attention to public mechanisms yields a strictly lower expected payoff to the designer.
\end{document}